%% file: brief.tex
\documentclass[journal,twoside,web]{ieeecolor} 
\usepackage{generic}
\usepackage{cite}
\usepackage{amsmath,amssymb,amsfonts}

\usepackage{amsthm} 
\usepackage{algorithmic}
\usepackage{graphicx}
\usepackage{textcomp}
\usepackage{comment}
\newtheorem{theorem}{Theorem}
\newtheorem{assumption}{Assumption}
\newtheorem{definition}{Definition}
\newtheorem{remark}{Remark}
\usepackage{soul}
\newtheorem{proposition}{Proposition}
\newtheorem{lemma}{Lemma}
\newtheorem{corollary}{Corollary}
\newtheorem{example}{Example}
\newcommand{\dn}{\mathbf{d}}
\newcommand{\ds}{\mathfrak{d}}
\newcommand{\R}{\mathbb{R}}

\usepackage{caption}
\usepackage{subcaption}
\usepackage{hyperref}

\def\BibTeX{{\rm B\kern-.05em{\sc i\kern-.025em b}\kern-.08em
    T\kern-.1667em\lower.7ex\hbox{E}\kern-.125emX}}
\begin{document}
\title{Discrete Homogeneity  and Quantizer Design for Nonlinear Homogeneous Control Systems}
\author{Yu Zhou,  Andrey Polyakov,  Gang Zheng,  and Masaaki Nagahara %\IEEEmembership{Member, IEEE}
\thanks{This work has been submitted to the IEEE for possible publication.
Copyright may be transferred without notice, after which this version may no longer be accessible.}
\thanks{This work is supported by JST ASPIRE Project Grant Number JPMJAP2402.}
\thanks{  Yu Zhou and Masaaki Nagahara are with the Graduate School of Advanced Science and Engineering, Hiroshima University, Hiroshima, Japan (e-mail:yuzhou@hiroshima-u.ac.jp, nagam@hiroshima-u.ac.jp). }
\thanks{Andrey Polyakov and Gang Zheng are with INRIA Lille, Centrale Lille, Lille University, Lille, France. (e-mail: andrey.polyakov@inria.fr, gangzheng@inria.fr).}
}

\maketitle

\begin{abstract}
This paper proposes a framework for analysis of generalized homogeneous control systems under state quantization. In particular, it addresses the challenge of maintaining finite/fixed-time stability of nonlinear systems in the presence of quantized measurements. To analyze the behavior of quantized control system, we introduce a new type of discrete homogeneity, where the dilation is defined by a discrete group.  The converse Lyapunov function theorem is established for homogeneous systems with respect to discrete dilations.
By extending the notion of sector-boundedness to a homogeneous vector space, we derive a generalized homogeneous sector-boundedness condition that guarantees finite/fixed-time stability of nonlinear control system under quantized measurements. A geometry-aware homogeneous static vector quantizer is then designed using generalized homogeneous coordinates, enabling an efficient quantization scheme. The resulting homogeneous control system with the proposed quantizer is proven to be homogeneous with respect to discrete dilation and globally finite-time, nearly fixed-time, or exponentially stable, depending on the homogeneity degree. Numerical examples validate the effectiveness of the proposed approach.
\end{abstract}


\section{Introduction}
\label{sec:introduction}

\input{introduction.tex}

\textit{Notations}: $\mathbb{R}$ is the set of real numbers; 
$\mathbb{R}_{+} = \{ x \in \mathbb{R} : x \ge 0 \}$;
$\mathbb{N}$ is the set of all natural numbers without zero; 
$\mathbb{Z}$ is the set of all integers; 
$\overline{1,n}$ denotes the index set $\{1,2,\dots,n\}$ for $n \in \mathbb{N}$; 
$\boldsymbol{0}$ denotes the zero vector in $\mathbb{R}^n$;  
$P \succ 0$ ($\prec 0$, $\succeq 0$, $\preceq 0$) for $P \in \mathbb{R}^{n \times n}$ means that $P$ is symmetric and positive (negative) definite (semidefinite); 
$\lambda_{\min}(P)$ and $\lambda_{\max}(P)$ represent the minimal and maximal eigenvalue of a symmetric matrix $P=P^\top$; 
for $P \succeq 0$, the square root of $P$ is a matrix $P^{1/2}$ such that $(P^{1/2})^2 = P$; 
$\|x\|=\sqrt{x^\top P x}$ is the weighted Euclidean norm of $x \in \mathbb{R}^n$, where $P \succ 0$ is defined dependently of the context. $\|A\|=\sup_{x\neq \mathbf{0}}\frac{\|Ax\|}{\|x\|}$ for $A\in \R^{n\times n}$.
 $\mathcal{K}$ is the set of continuous strictly increasing functions $\sigma: \mathbb{R}_+\mapsto\mathbb{R}_+$ such that $\sigma(0)=0$; $I_n$ denotes the identity matrix in $\mathbb{R}^n$. For $a \in \mathbb{R}$, we denote its lower integer part by $\lfloor a \rfloor$ and its upper integer part by $\lceil a \rceil$.  $f \circ g$: composition of $f$ and $g$, defined as $(f \circ g)(x) = f(g(x))$.

\section{Preliminaries on homogeneity}\label{section:3}
\input{preliminary.tex}

\section{Problem Formulation}\label{sec:problem}
\input{problem_statement}

\section{Homogeneity with discrete dilation}\label{section:discrete}
\input{discrete_homogeneity}

\section{Homogeneous sector-boundedness}\label{section:hom_sector}
\input{sector_boundedness}

\input{main_results.tex}

\section{Numerical example}\label{section:example}
\input{simulation}

\section{Conclusion}
We have developed a framework to study scaling symmetries of nonlinear homogeneous systems and finite/fixed-time properties under state quantization. In this context, the notion of a discrete dilation was introduced. It  is shown that an asymptotically stable homogeneous system with respect to a discrete dilation group still admits a smooth homogeneous Lyapunov function with the same dilation, and  its convergence (finite-time, nearly fixed-time, or exponential) rate depends on the homogeneity degree. The concept of homogeneous sector-boundedness for homogeneous vector spaces enables the formulation of stability conditions for homogeneous control under quantized data in terms of sector-bounded quantization error. Furthermore, when the system is homogeneous with respect to a discrete dilation, these stability conditions need only be verified over a compact set.
A geometry-aware approach for designing homogeneous polar-spherical quantizers for nonlinear homogeneous systems was presented. By exploiting the system’s inherent homogeneity, the quantizer ensures that the closed-loop dynamics remain discretely homogeneous and satisfy the homogeneous sector-boundedness condition. Future work may include extending this framework to observer design and to applications in robotic systems, with an emphasis on guaranteeing safety under coarse quantized measurements.

\section*{References}
\bibliographystyle{unsrt}
\bibliography{reference}

\appendix
\input{appendix}

\end{document}

%% file: introduction.tex
Quantization is crucial in digital control and has also gained substantial interest in networked control for reducing data transmission. A quantizer is a mapping from a continuous state space to a discrete set of admissible values.  Quantizers are utilized in control systems, e.g., for transmission of  analog sensor measurements to a digital controller.  Quantizers can be classified into two categories: static (time-invariant) quantizers and dynamic (time-varying) quantizers.  In the latter case, the parameters of the quantizer vary in time. In this study, we focus on control systems with static quantizers. 

One of the most important concerns in control with quantization is the \textit{stability}.
In the context of nonlinear stabilization with a static quantizer, most studies have been focused on the stability and robustness of the system with conventional uniform and logarithmic quantizers (see e.g., \cite{liu_jiang_etas2012Aut},\cite{bikas2020:TAC}, \cite{wang2021:TAC}, \cite{ceragioli2007:SCL}). However, \textit{the structure of logarithmic and uniform quantizers does not take into account any information about nonlinear control system.} As the result, asymptotic stability cannot be generally guaranteed for closed-loop nonlinear systems. Uniform quantization introduces nonvanishing errors and typically ensures only practical stability. The logarithmic quantizer, designed based on quadratic Lyapunov functions for linear systems \cite{elia_etal_2001_TAC}, leads to sector-bounded quantization errors \cite{Fu_etal_2005_TAC}, which may cause instability in nonlinear settings. This highlights the need for model-based quantizer design that explicitly incorporates nonlinear dynamics. Despite its importance, static quantizer design for nonlinear systems remains largely unresolved.

Another important consideration, that is often overlooked in the design of quantizers for control systems, is \textit{quantization efficiency}. Quantization efficiency measures how well a continuous state space can be captured using a countable number of quantization levels. Most existing approaches rely on quantizers that independently quantize each component of the state vector. While this may be sufficient for low-dimensional systems, it becomes inefficient in higher-dimensional settings, leading to a redundant use of quantization levels and increased computational burden. 

Vector quantization, which jointly quantizes multiple components, offers a more efficient encoding-decoding algporithms. One of the most popular vector quantizers  is the polar-spherical quantizer, which uses polar-spherical coordinates. This method was first introduced in communication theory (see, e.g., \cite{swaszek1983multidimensional}) and later implemented for feedback stabilization of control systems \cite{bullo_Liberzon2006TAC}. It decomposes a vector into its norm and direction (on the unit sphere). This decomposition enables more efficient quantization, as the unit sphere can be covered by a finite number of regions (finite subcovers).
Polar-spherical quantizers (see Fig~\ref{fig:SISO_Vector} for an illustration) have been applied in \cite{gu2014SCL}, \cite{wang2018Aut}, and \cite{wang2021TAC} for the stabilization of linear systems. \textit{However, to the best of the authors' knowledge, the design and application of a vector quantizer for nonlinear control systems remains an open research problem.}

To address these gaps, it is important to incorporate a geometric structure of nonlinear dynamics into an algorithm of space partition and a design of the quantizer. Using geometric properties (e.g., symmetry, invariance, foliation), quantizers can be constructed in alignment with the structure of the system.

Among various geometric structures, homogeneity (dilation symmetry) plays an important role in control systems design \cite{polyakov2025:book_vol_I},\cite{polyakov2025:book_vol_II}. A homogeneous control system may have a better regulation quality comparing with linear algorithm. In this paper, we argue that exploiting homogeneity provides an effective approach to address two key challenges in quantized control: achieving \textit{finite/fixed-time}  stabilization\footnote{
Finite-time stabilization means the state converges to the equilibrium in a finite time, $T(x_0)$, where the settling time is dependent on the initial state $x_0$ \cite{bhat2000:SIAM}. Fixed-time stabilization is the stronger property in which the settling time $T_{\max}$ is uniformly bounded and independent of $x_0$ for all initial conditions \cite{polyakov2011:TAC}.} of nonlinear control systems and improving \textit{quantization efficiency} beyond conventional methods of quantizers design.

Homogeneous systems naturally arise in the approximation of nonlinear dynamics. As a special class of nonlinear systems, they retain several key properties of linear systems, making them particularly attractive for control and quantizer design. 
The homogeneity provides several important properties for control system design and analysis: the existence of homogeneous Lyapunov function, equivalence of local and global stability, and the ability to tune finite/fixed-time convergence by homogeneity degree (see, e.g., \cite{zubov1958systems}, \cite{khomenyuk1961systems}, \cite{rosier1992SCL}, \cite{kawski1990_CTAT}, \cite{bhat_etal_2005_MCSS}). Over the past two decades, various homogeneous control strategies have been developed for continuous-time systems (e.g., \cite{grune2000SIAM}, \cite{hong2001output}, \cite{andrieu2008homogeneous}, \cite{nakamura_etal_2009_TAC}, \cite{polyakov2019IJRNC}, \cite{zimenko_etal_2020_TAC}).

However, quantization may disrupt the desirable properties of homogeneous systems. A homogeneous control approach with quantized data was proposed in \cite{zhou2024:Tac} for linear systems, where a component-wise logarithmic quantizer and a diagonalizable dilation were used to achieve finite/fixed-time stabilization. For general nonlinear homogeneous systems with non-diagonalizable dilations, the problem of designing a quantizer that preserves finite/fixed-time stability is still unsolved.

{
Conventionally, homogeneity is defined with respect to a continuous dilation (scaling) group, which assumes the scaling parameter spans the entire real line. While this formulation provides a powerful analytical framework, it represents an idealized view that is not applicable to  digital control systems. In reality, with digital measurements, quantized sensing, or sampled-data controllers, the available information is inherently discrete, and the system's scaling behavior may be manifested only along a discrete set of states or time instants.

For scenarios involving time-induced discreteness, the concept of discrete homogeneity was first introduced as D-homogeneity in \cite{sanchez2020:Aut}. This concept was subsequently utilized for discrete-time control design in \cite{granzotto2021:CDC}. Building upon this work, the new notion of S-homogeneity for discrete-time systems was proposed in \cite{grune2023:ifac}, providing a framework that enables stability properties (e.g., practical and local) to be derived from the homogeneity degree. While these studies have established important stability and robustness results, the key limitation of all these approaches is that the group of scaling transformations utilized in a definition of the homogeneity remains continuous.

In contrast, when discreteness arises from quantization, the state space is partitioned into disjoint subsets, and only quantized values are available. This fundamentally breaks the classical notion of continuous scaling symmetry. This work, therefore, addresses two fundamental questions in homogeneity theory:
1) \textit{Can a meaningful notion of homogeneity be introduced for systems with quantization-induced discreteness?} 2) \textit{How can this property be exploited for the analysis and design of feedback controllers for systems with quantized states?}
}

The paper also addresses the problem of ensuring stability for nonlinear homogeneous systems with state quantization. We provide a comprehensive framework that incorporates \textit{a new theoretical notion of discrete homogeneity,  a finite/fixed-time stability analysis and a geometry-aware quantizer design  preserving the homogeneity of the control system}. The main contributions are as follows:

    { 1) \textit{Discrete homogeneity}: We introduce the concept of {discrete homogeneity}, where the scaling transformation is defined by a discrete group. 
    This type of homogeneity provides a consistent theoretical model for homogeneous systems subject to space discreteness (e.g., quantization). We demonstrate that discrete homogeneity preserves many favorable properties of continuous homogeneity, including the existence of homogeneous Lyapunov functions and characterization of finite/fixed-time convergence by homogeneity degree.
}
     
     2) \textit{Homogeneous sector-boundedness:} We introduce a generalized notion of sector-boundedness formulated within a homogeneous vector space. This concept extends the classical sector bounds to systems that are symmetric with respect to generalized dilation (scaling), providing a foundation for various methods of homogeneity-based analysis. 
    
    3) \textit{Finite/fixed-time stability analysis of nonlinear systems with state quantization}: We establish that finite-/fixed-time stability of generalized homogeneous systems can be guaranteed if the quantization error satisfies a homogeneous sector-bounded condition. Moreover, due to discrete homogeneity, the global stability analysis can be reduced to an analysis of the nonlinear system on a compact set.
     
     4) \textit{Homogeneous polar-spherical quantizer:}  Leveraging system homogeneity, we construct a vector quantizer based on homogeneous polar-spherical coordinates. Importantly, the resulting closed-loop system remains homogeneous in the discrete sense. The induced quantization error satisfies a homogeneous sector-bounded condition under any linear dilation. This approach provides a geometry-aware quantizer design that preserves stability, homogeneity, and quantization efficiency, in contrast to conventional component-wise quantization schemes.

	The paper is organized as follows. Section~\ref{section:3} provides a brief introduction to homogeneity with respect to a continuous dilation, homogeneous systems and functions, along with some useful preliminary results.
Section~\ref{sec:problem} presents a problem statement for nonlinear homogeneous control systems with quantized measurements.
Section~\ref{section:discrete} establishes discrete homogeneity and presents some theorems on homogeneous dynamics with respect to discrete dilation.
Section~\ref{section:hom_sector} introduces the notion of homogeneous sector-boundedness and discusses its key properties.
Section~\ref{sec:condition} presents a general stability condition for homogeneous control systems under state quantization.
Section~\ref{sec:quantizer} details the design of the homogeneous polar-spherical quantizer.
A numerical example is given in Section~\ref{section:example}.

%% file: preliminary.tex
Homogeneity refers to a class of dilation symmetries, which have been shown to possess several useful properties for control design and analysis \cite{zubov1958systems}, \cite{khomenyuk1961systems}, \cite{kawski1991}, \cite{rosier1992SCL},  \cite{bhat_etal_2005_MCSS}, \cite[Chapter 1]{polyakov2020book}. 

\begin{definition}
	\textit{A one-parameter family of mappings $\dn(s):\mathbb{R}^n\mapsto \mathbb{R}^n$ with $s\in\mathbb{R}$ is said to be a  dilation in $\mathbb{R}^n$ if} 
	\begin{itemize}
		\item $\dn(0) \!=\! I_n$, $\dn(s+t) \!=\! \dn(s)\circ \dn(t)=\dn(t)\circ \dn(s)$, $\forall s, t\!\in\!\mathbb{R}$;
		\item  $\lim\limits_{s\rightarrow -\infty}\|\dn(s)x\| =0$ and $\lim\limits_{s\rightarrow \infty}\|\dn(s)x\| =\infty$,  $\forall x\neq \mathbf{0}$.  
	\end{itemize}
\end{definition}

In this paper, we deal with the one-parameter group of linear continuous dilations (\textit{linear dilation}) that can be defined as
\[
\dn(s) = e^{sG_\dn}:=\sum_{i=0}^{\infty} \frac{s^i G_\dn^i}{i!}, \quad s\in \R,
\]
where an anti-Hurwitz  matrix $G_\dn\in \R^{n\times n}$ is the generator of the dilation satisfying $\tfrac{d}{ds}e^{G_\dn s} = e^{G_\dn s} G_\dn= G_\dn e^{G_\dn s}$, $\forall s\in \R$.

\begin{definition}
	A dilation $\dn$ is strictly monotone with respect to a norm $\|\cdot\|$ in $\R^n$ if $\exists \beta>0$ such that
	$
	\|\mathbf{d}(s)\| \leq e^{\beta s}, \forall s \leq 0$.
\end{definition}
The following result is the straightforward consequence of the quadratic Lyapunov function theorem for linear systems. 
\begin{proposition}\cite{polyakov2019IJRNC}\label{prop: 1}
	A linear continuous dilation $\mathbf{d}$ is strictly monotone  with respect to the weighted Euclidean norm $\|z\|=\sqrt{z^\top P z}$ if and only if the following linear matrix inequality holds
	$ P\succ 0$, 
	$P G_\dn+G_\dn^{\top} P \succ 0
	$,
	where $G_\dn \in \mathbb{R}^{n\times}$ is the generator of $\mathbf{d}$
    .
\end{proposition}

The homogeneous function and vector field are defined following the papers  \cite{zubov1958systems},\cite{kawski1991}.
\begin{definition}\label{def:hom_sys}
	A  vector field $f:\R^n\mapsto \R^n$ (resp., a function $h:\R^n \!\mapsto\! \R$) is said to be $\dn$-homogeneous of degree $\mu\!\in\! \R$ if
	\[
	f(\dn(s)x)=e^{\mu s} \dn(s) f(x), \ (\text{resp., } h(\dn(s)x)=e^{\mu s} h(x)), 
	\]
	$ \forall x\in\R^n$, $\forall s\in \R,$ where $\dn$ is a linear continuous dilation.
\end{definition}
	If a homogeneous mapping is smooth, then its derivative is homogeneous as well \cite[Propotion 7.4]{polyakov2020book}, \cite[Corollary 2]{polyakov2019IJRNC}.
	\begin{proposition}\label{prop:partial}
		Let  a function $h\in C^1(\mathbb{R}^n\setminus\{\boldsymbol{0}\}, \mathbb{R}$) and a vector field $g\in C^1(\mathbb{R}^n\setminus\{\boldsymbol{0}\}, \mathbb{R}^n)$ be $\dn$-homogeneous of degree $\mu\in \R$, then 
		\begin{equation}
		\left.\frac{\partial h(z)}{\partial z}\right|_{z=\dn(s)x}\dn(s)=	e^{\mu s} \frac{\partial h}{\partial x} ,\;\;\;\; \frac{\partial h(x)}{\partial x} G_\dn x = \mu h(x),
		\end{equation}
	\begin{equation}\label{eq:partial_map}
		\left.\frac{\partial g(z)}{\partial z}\right|_{z=\dn(s)x}\dn(s) =e^{\mu s}\dn(s)\frac{\partial g(x)}{\partial x},
	\end{equation}
		for all $x\in\mathbb{R}^n\setminus \{\boldsymbol{0}\}$ and $s\in\mathbb{R}$.
	\end{proposition}
The linear continuous dilation $\dn$ induces an alternative topology in $\mathbb{R}^n$ via a ``homogeneous norm'' \cite{kawski1995IFAC}.

\begin{definition}\cite{polyakov2019IJRNC}\label{def:hom_norm}
	The function $\|\cdot\|_\dn: \mathbb{R}^n \mapsto\R_+$ defined as $\|\mathbf{0}\|_{\dn}=0$ and 
	$
\|x\|_{\mathbf{d}}=e^{s}$, where  $s \in \mathbb{R}:\left\|\mathbf{d}\left(-s\right) x\right\|=1, x\neq \mathbf{0}
	$,
	is called the canonical homogeneous norm in $\mathbb{R}^n$, where $\mathbf{d}$ is a continuous linear dilation being monotone with respect to the norm $\|\cdot\|$ in $\mathbb{R}^n$.
\end{definition}

The monotonicity of the dilation is required to guarantee that the functional
$\|\cdot\|_\dn$ is single-valued and continuous at the origin \cite{polyakov2019IJRNC}, \cite[Corollary 6.4]{polyakov2020book}.

\begin{proposition}\label{prop:diation_inequality}
	Let $\mathbf{d}$ be a strictly monotone linear continuous dilation on  $\mathbb{R}^n$. Then 
	\begin{equation}\label{eq:d_bound}
		\left\{
		\begin{aligned}
			&e^{\underline{\eta} s} \leq|\lfloor\dn(s)\rfloor|\le\|\mathbf{d}(s) \| \leq e^{\overline{\eta} s}, \ & s \geq 0, \\
			&e^{\overline{\eta} s} \leq |\lfloor\dn(s)\rfloor|\le \|\mathbf{d}(s) \| \leq e^{\underline{\eta} s}, \ &  s \leq 0,
		\end{aligned}
		\right. \quad \forall s\in \R,
	\end{equation}
	\begin{equation}\label{eq:hom_eclidean}
		\left\{
		\begin{aligned}
			\begin{aligned}
				& \|x\|_\dn^{\underline{\eta}} \leq\|x\| \leq \|x\|_\dn^{\overline{\eta}},&\ \|x\|\ge 1,\\
				& \|x\|_\dn^{\overline{\eta}} \leq\|x\| \leq \|x\|_\dn^{\underline{\eta}}, &\ \|x\|\le 1 ,	
			\end{aligned}
		\end{aligned}
		\right. \quad \forall x \in \mathbb{R}^n,
	\end{equation}
	where $|\lfloor\dn(s)\rfloor|=\inf _{u \in S}\|\mathbf{d}(s) u\|=\inf _{u \neq 0} \frac{\|\mathbf{d}(s) u\|}{\|u\|}$, \\
	$
	\overline{\eta}=\tfrac{1}{2} \lambda_{\max }(P^{\frac{1}{2}} G_{\mathrm{d}} P^{-\frac{1}{2}}+P^{-\frac{1}{2}} G_{\mathrm{d}}^{\top} P^{\frac{1}{2}})>0,
	$\\
	$	\underline{\eta}=\tfrac{1}{2} \lambda_{\min }(P^{\frac{1}{2}} G_{\mathrm{d}} P^{-\frac{1}{2}}+P^{-\frac{1}{2}} G_{\mathrm{d}}^{\top} P^{\frac{1}{2}})>0.
	$
\end{proposition}

 The following result is the straightforward corollary of Zubov--Rosier theorem on homogeneous Lyapunov function for asymptotically stable homogeneous system \cite{zubov1958systems}, \cite{rosier1992SCL}, \cite{nakamura2002:SICE}.

\begin{theorem}\label{thm:exist_LF}
    Let vector field $f:\R^n\mapsto \R^n$ be $\dn$-homogeneous  of degree $\mu\in \R$. The system 
\begin{equation}\label{eq:f(x)} 
		\dot{x} = f(x)
	\end{equation}
	is globally uniformly asymptotically stable if and only if there exists a positive definite
	$\dn$-homogeneous function $V:\R^n\to [0,+\infty)$ of degree $m>0$ such that $V\in C(\R^n)\cap C^1(\R^n\backslash\{\boldsymbol{0}\}),$
	\[\dot V(x)\leq -\rho V^{1+\tfrac{\mu}{m}}(x), \quad \forall x\neq \boldsymbol{0}.\]
    where $\rho>0$ is some number. Moreover, the system is 
    \begin{itemize}
        \item globally uniformly finite-time stable\footnote{The system \eqref{eq:f(x)} is finite-time stable if it is Lyapunov stable and $\exists T(x_0): \|x(t)\|=0, \forall t\geq T(x_0), \forall x_0\in \R^n$.} for $\mu<0$; 
        \item globally uniformly exponentially stable for $\mu=0$; \item globally uniformly nearly fixed-time stable\footnote{The system \eqref{eq:f(x)} is uniformly nearly fixed-time stable it is Lyapunov stable and $\forall r>0, \exists T_r>0: \|x(t)\|<r, \forall t\geq T_r$ independently of $x_0\in \R^n$.} for $\mu>0$. 
    \end{itemize}
\end{theorem}

\begin{proposition}\label{prop: V_V}
\label{thm:F_1_2}\cite{bhat_etal_2005_MCSS}
Suppose $F_1$ and $F_2$ are continuous real-valued function on $\mathbb{R}^n$, $\dn$-homogeneous of degree $\nu_1>0$ and $\nu_2>0$, respectively, and $F_1$ is positive definite. Then, for every $x\in\mathbb{R}^n$, \[
c_{\min} F_1^{\nu_2/\nu_1}(x) \le F_2(x) \le c_{\max} F_1^{\nu_2/\nu_1}(x),\]
where $c_{\min} := \min_{F_1(z) = 1} F_2(z)$, $c_{\max} := \max_{F_1(z) = 1} F_2(z)$.
\end{proposition}

%% file: problem_statement.tex
We consider a class of continuous-time closed-loop control systems of the form
\begin{equation}\label{eq:f}
    \dot{x} = f(x) + g(x) u(x),
\end{equation}
where $x \in \mathbb{R}^n$ is the state vector, $u(x) \in \mathbb{R}^m$ is the state-feedback control, the vector fields $f:\mathbb{R}^n \mapsto \mathbb{R}^n$, $g:\mathbb{R}^n \mapsto \mathbb{R}^{n \times m}$ and the feedback law $u:\mathbb{R}^n \mapsto \mathbb{R}^m$ are assumed to be continuous.
\begin{assumption}\label{assmp:1}
	The origin of the closed-loop system \eqref{eq:f} is globally asymptotically stable. 
\end{assumption} 

\begin{assumption}\label{assmp:2}
    The closed-loop system \eqref{eq:f} is homogeneous of degree $\mu \in \mathbb{R}$ with respect to a continuous linear dilation $\dn:\mathbb{R}\mapsto \mathbb{R}^{n \times n}$ with a generator $G_\dn\in \mathbb{R}^{n\times n}$.
\end{assumption}
We assume that the state measurement is quantized. 
Mathematically, this means that the state $x$ in the control $u(x)$ is replaced by a quantized state  $\mathfrak{q}(x)$, where 
$\mathfrak{q}: \mathbb{R}^n\mapsto \mathcal{Q}\subset \mathbb{R}^n
$ 
is a discrete function $q:\R^n\mapsto \mathcal{Q}$ that maps disjoint subsets $\mathcal{D}_i\subset\R^n$ to vectors $q_i\in \mathcal{D}_i$ as follows:
$
\mathfrak{q}(x) = \mathfrak{q}_i, \forall x\in \mathcal{D}_i, 
$ where $i\in \mathbb{N}$,  $\mathcal{Q}:=\cup \mathfrak{q}_i\subset \mathbb{R}^n$  is 
a countable discrete set and $\cup \mathcal{D}_i = \mathbb{R}^n$.
The vector $\mathfrak{q}_i$ is called the \textit{quantization seed} of the  \textit{quantization cell} $\mathcal{D}_i$.

The control system \eqref{eq:f} with quantized state measurements can be represented as:
\begin{equation}\label{eq:q_sys}
	\dot{x} = f(x) + g(x)u(\mathfrak{q}(x)),
\end{equation}
where, by assumption, the quantizer $\mathfrak{q}$ is such that $\mathfrak{q}(\boldsymbol{0}) = \boldsymbol{0}$.
The above system has a discontinuous right-hand side. Its solutions are understood in the sense of Filippov \cite{filippov2013differential}.

In the general case, the state quantization breaks continuous homogeneity, i.e.,
$f(\dn(s)x) + g(\dn(s)x)u( \mathfrak{q}(\dn(s)x)) \neq
e^{-\mu s} \dn(s) \big[ f(x) + g(x)u (\mathfrak{q}(x)) \big]$,
 for some $x\in \R^n$ and some $s\in \R$.
 %that is, the quantized homogeneous system may be non-homogeneous in the classical sense. This observation shows that . 
 Therefore, new analysis and design tools are needed to be developed in order to capture a dilation symmetry of the system in the case of state quantization. The objectives of this paper are as follows: 
\begin{itemize}
\item  Develop a concept of homogeneity for quantized systems.
    \item Establish conditions under which the quantizer $\mathfrak{q}$ preserves the global asymptotic stability of the  system~\eqref{eq:q_sys}.
\item Design a homogeneity-based vector quantizer that preserves a homogeneity and asymptotic stability of the quantized homogeneous control system \eqref{eq:q_sys}.
\end{itemize}

Various state quantizers are developed for linear control systems design and analysis. For example,
the element-wise logarithmic quantizer \cite{elia_etal_2001_TAC, Fu_etal_2005_TAC} (see the left side of Fig. \ref{fig:SISO_Vector}) has an infinite number of quantization levels as each coordinate approaches zero.
The polar-spherical vector quantizer \cite{wang2018Aut, wang2021TAC}   has an infinite number of levels only when the state approaches the origin (see the right side of Fig. \ref{fig:SISO_Vector}). 
For any compact set $K_{r_1,r_2}=\{(x,y)\in \R^2:   0<r_1 \le x^2+y^2 \le r_2\}$, the element-wise quantizer has an infinite number of quantization seeds/cells, 
while the number of quantization seeds/cells for the polar-spherical quantizer on this set is always finite.
In this paper, we are particularly aimed at the development of an analog of the polar-spherical quantizer for homogeneous control systems.
\begin{figure}[ht]
    \centering
    \includegraphics[width=1\linewidth]{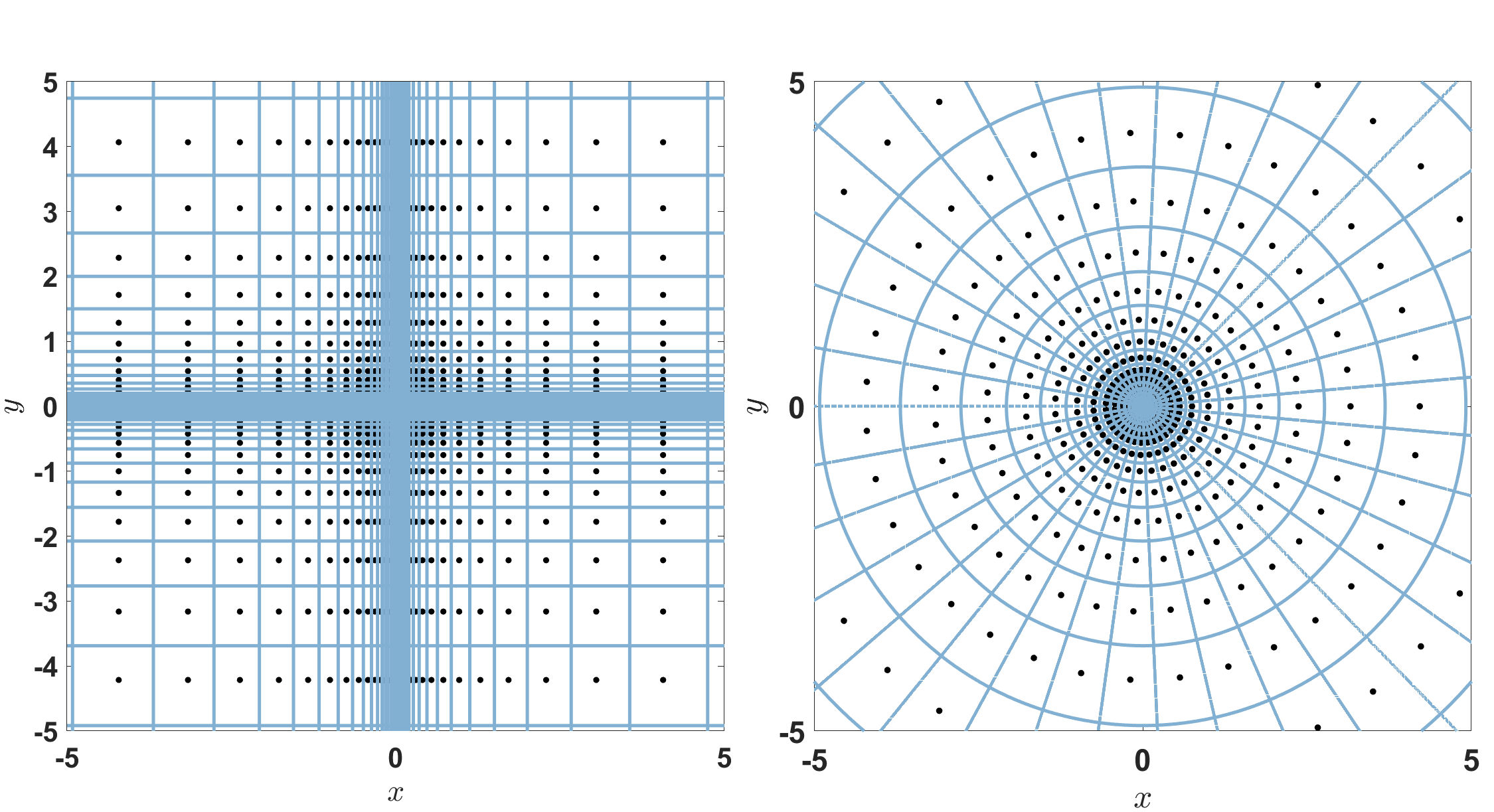}
\caption{
Comparison between element-wise logarithmic (left) and polar-spherical (right) quantizers, where the black points represent quantization seeds. 
}
    \label{fig:SISO_Vector}
\end{figure}

%% file: discrete_homogeneity.tex
We extend the concept of continuous dilations to the discrete case. 
In this case, the dilation parameter takes values in a discrete additive subgroup of $\mathbb{R}$. 
This extension describes systems that are homogeneous under discrete scalings. In Section VII we demonstrate that a discrete homogeneity occurs, for instance, in systems with quantized measurements.
\begin{definition}
	A a one-parameter family of mappings $\mathfrak{d}(s):\mathbb{R}^n \mapsto \mathbb{R}^n$, $s \!\in\! \mathcal{S} \!\subset\! \mathbb{R}$ is said to be a \textit{discrete dilation} in $\mathbb{R}^n$ if  
	\begin{itemize}
		\item a countable set $\mathcal{S}$  with $0\in \mathcal{S}$ is an additive subgroup of $\mathbb{R}$ that is unbounded in both directions, i.e., $\sup \mathcal{S} = +\infty$ and $\inf \mathcal{S} = -\infty$;
		\item $\mathfrak{d}(0) = I_n$, $\mathfrak{d}(s+t) = \mathfrak{d}(s)\circ \mathfrak{d}(t) = \mathfrak{d}(t)\circ \mathfrak{d}(s)$, $\forall s,t \in \mathcal{S}$;
		\item for any $x \neq 0$, $\|\mathfrak{d}(s)x\| \to 0$ as $s \to -\infty$ in $\mathcal{S}$ and $\|\mathfrak{d}(s)x\| \to \infty$ as $s \to +\infty$ in $\mathcal{S}$.
	\end{itemize}
     A discrete dilation $\mathfrak{d}$ is linear if $\mathfrak{d}(s)\in \R^{n\times n}$ for all $s\in\mathcal{S}$.
\end{definition}
 A (conventional) dilation $\dn(s)$ with $s\in \R$  can be transformed into a discrete dilation by a proper partition  of
the real line~$\mathbb{R}$  into a discrete set $\mathcal{S}$. 
This partition must be carefully designed so that it preserves the additive structure of the underlying subgroup. 
    For example, the following partition
    $ \{ a^i \mid i \in \mathbb{Z} \}, a>1
    $ is \emph{not additive}, since for example $a^1 + a^1 = 2a \notin \mathcal{S}_{a}$. Therefore, it cannot directly define a discrete dilation.  However, taking the logarithm of the partition,
    $\{ i \ln a \mid i \in \mathbb{Z} \},
    $ 
    yields an additive subgroup of $\mathbb{R}$, since for any $i,j \in \mathbb{Z}$,
    $
        i \ln a + j \ln a = (i+j) \ln a.
    $ 
    Thus, the log-based partition preserves the additive structure required for discrete dilation. 
    
    Indeed, it is well known that any discrete additive subgroup of~$\mathbb{R}$ 
is cyclic~\cite{rudin1976:book}; that is, the discrete set $\mathcal{S}$ for discrete dilation can always be represented as  
\begin{equation}\label{eq:S}
    \mathcal{S} := \{\, k a \mid k \in \mathbb{Z},\ a > 0 \,\}.
\end{equation}
We refer to the elements $ka$ of this set as \textit{seeds of a discrete dilation}.
The discrete set $\mathcal{S}$ becomes dense in $\mathbb{R}$ as $a\to 0$.
Therefore,  the linear discrete dilation $\ds$ recovers (in some sense) the linear continuous dilation $\dn$ as $a\to 0$.

 Since our study begins with a well-designed homogeneous system with respect to a linear continuous dilation, in this paper, we deal only with  linear discrete dilations induced by linear continuous dilations. More precisely, we consider \textit{a linear discrete dilation defined as
	\[
	\ds(s) = e^{s G_\dn} := \sum_{i=0}^{\infty} \frac{s^i G_\dn^i}{i!}, \quad s \in \mathcal{S} \subset \mathbb{R},
	\]
	where $\mathcal{S}\subset \R$ is given by \eqref{eq:S} and  $G_\dn \in \mathbb{R}^{n \times n}$ is an \textit{anti-Hurwitz} matrix being the \textit{generator} of the linear continuous dilation $\dn$.} All results given below are proven only for such a linear discrete dilation in $\R^n$.

Distinction between continuous and discrete dilations is:
\begin{enumerate}
    \item \textit{Continuous linear dilation}: Parameterized by a continuous scalar $s \in \mathbb{R}$, the dilation $\dn(s) = e^{sG_\dn}$ defines a one-parameter continuous group acting on the state space. The action $\dn(s)x$ traces a continuous orbit (\textit{continuous homogeneous curve}) through $\mathbb{R}^n \setminus \{\boldsymbol{0}\}$.
    \item \textit{Discrete linear dilation}: Parameterized by a discrete integer $k \in \mathbb{Z}$, the dilation $\mathfrak{d}(ka) = e^{kaG_{\dn}}$ defines a discrete subgroup. The action $\mathfrak{d}(ka)x$ generates a discrete orbit\footnote{a countable set of isolated points in $\mathbb{R}^n$} (\textit{discrete homogeneous curve}) corresponding to admissible scaling levels.
\end{enumerate}
Figure \ref{fig:discrete_seeds} illustrates this difference, showing the continuous path under continuous dilation versus the discrete scaling of a point under a discrete linear dilation, generating a non-uniform set of discrete scaling levels.
\begin{figure}[h]
    \centering
    \includegraphics[width=0.5\textwidth]{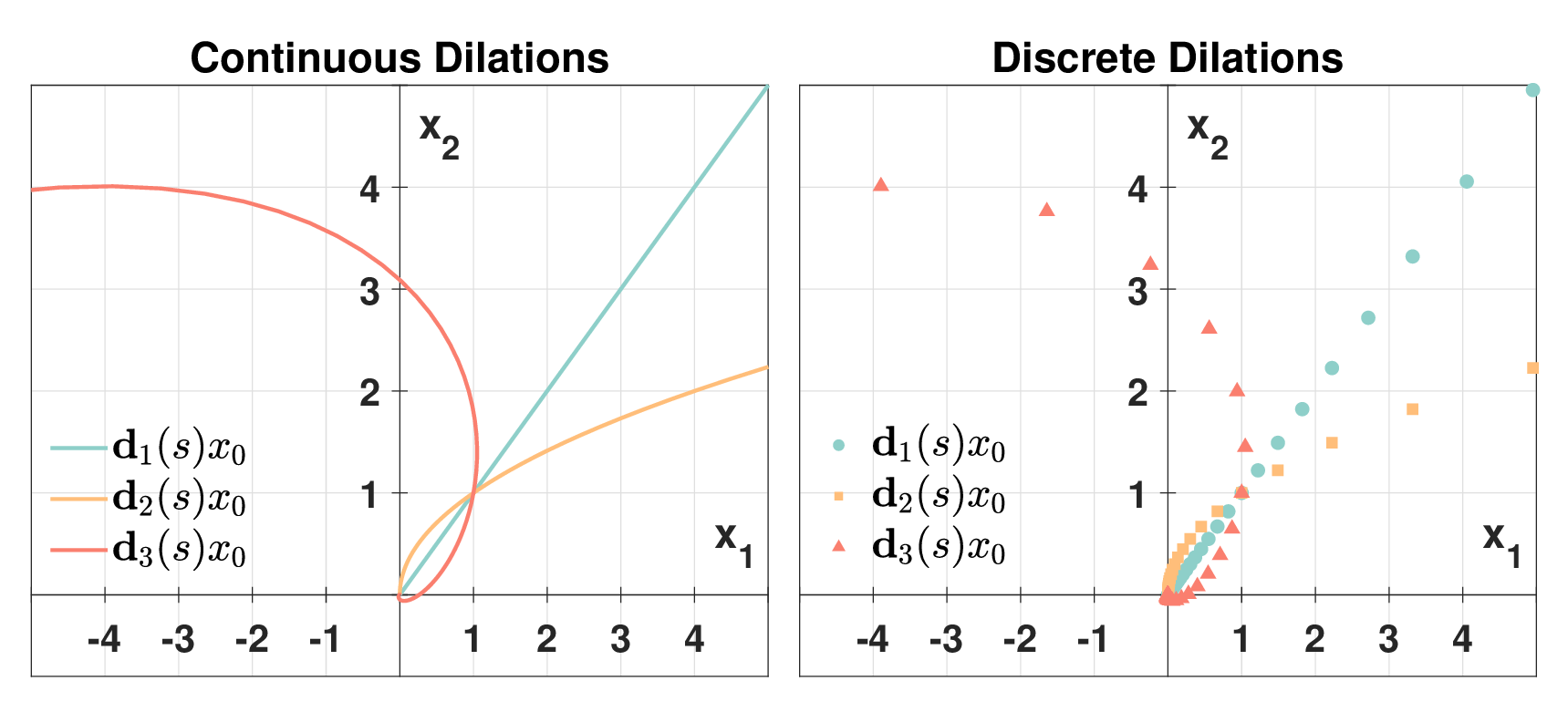}
    \caption{Homogeneous curves for  linear continuous and discrete dilations generated by    $ e^{G_{\dn_i} s}$, $i\!=\!1,2,3$, $G_{\dn_1}\!=\!I_2$, $G_{\dn_2}\!=\!\left[\begin{smallmatrix}
        2 & 0\\
        0 & 1
    \end{smallmatrix}\right]$, $G_{\dn_3}\!=\!\left[\begin{smallmatrix}
        2 & -1.5\\
        1 & 1
    \end{smallmatrix}\right],$ $x_0\!=\!\left[\begin{smallmatrix}
        1\\
        1
    \end{smallmatrix}\right],$ where $s\!\in\! \R$ (left) or $s
\!\in\!\mathcal{S}$ (right).} \label{fig:discrete_seeds}
\end{figure}

We now extend the concept of homogeneous mappings to the discrete dilation setting.

\begin{definition}
	A vector field $f:\mathbb{R}^n \mapsto \mathbb{R}^n$ (resp., a function $h:\mathbb{R}^n \mapsto \mathbb{R}$) is said to be  $\mathfrak{d}$-homogeneous of degree $\mu \in \mathbb{R}$ if
	\[
	f(\mathfrak{d}(s)x) = e^{\mu s} \mathfrak{d}(s) f(x), \quad (\text{resp., } h(\mathfrak{d}(s)x) = e^{\mu s} h(x)),
	\]
	for all $x \in \mathbb{R}^n$ and $s \in \mathcal{S}$, where $\ds$ is a linear discrete dilation.
\end{definition}
 The  $\ds$-homogeneity of a set-valued vector field $F:\mathbb{R}^n \rightrightarrows \mathbb{R}^n$ (resp., a set-valued function 
$H:\mathbb{R}^n \rightrightarrows \R$)  is defined by  the same formulas understood in set-theoretic sense. 
\begin{remark}
    Since $\mathcal{S}\subseteq \mathbb{R}$, then any $\dn$-homogenoeus vector field (function) is also $\mathfrak{d}$-homogeneous for same generator $G_{\dn}$.
Moreover, the $\mathfrak{d}$-homogeneity is preserved under addition and composition provided that at least one of the components is continuously $\mathfrak{d}$-homogeneous.
If a vector filed $f$ is $\mathfrak{d}$-homogeneous of degree $\mu$, a vector field $\tilde{f}$ is  $\mathfrak{d}$-homogeneous of degree $\mu$, and a mapping $\tilde{f}_1$ commutes with $\mathfrak{d}$, then
\begin{itemize}
\item $f + \tilde{f}$ is $\mathfrak{d}$-homogeneous of degree $\mu$;
\item $f \circ \tilde{f}_1$ is $\mathfrak{d}$-homogeneous of degree $\mu$.
\end{itemize}
\end{remark}

In the continuous dilation setting, any nonzero vector can be uniquely projected onto the unit sphere by using homogeneous norm (i.e., $\|\dn(-\ln\|x\|_\dn)x\|=1$). In the discrete dilation setting, any nonzero vector can be uniquely projected onto a compact set (a "donate" containing the sphere).

\begin{lemma}
Let $\dn$ be a linear continuous dilation in $\R^n$ 
and $\ds$ be the correspong linear discrete dilation in $\R^n$
with $\mathcal{S}\subset\R$ be given by \eqref{eq:S}.
Let $\|\cdot\|_{\dn}$ be the canonical $\dn$-homogeneous norm induced by the weighted Euclidean norm $\|x\| = \sqrt{x^{\top} P x}$ with $P \succ 0$.  
Let us define the set
\[
\Omega_{\mathfrak{d}}(\varrho) := \{\, z \in \mathbb{R}^n \setminus \{\boldsymbol{0}\} \mid \varrho \le \|z\|_{\dn} < \varrho e^{a} \,\}, \quad \rho>0.
\]
Then there exists the unique  function $k_{\mathfrak{d}}^{\varrho}: \mathbb{R}^n \setminus \{\boldsymbol{0}\} \!\mapsto\! \mathbb{Z}\!$ given by $k_{\mathfrak{d}}^{\varrho}(x)=\lceil a^{-1}\ln \frac{\rho}{\|x\|_{\dn}}\rceil$ such that 
\[
\mathfrak{d}\left(-k_{\mathfrak{d}}^{\varrho}(x) a\right) x \in \Omega_{\mathfrak{d}}(\varrho)
\]
for all  $x \in \mathbb{R}^n \setminus \{\boldsymbol{0}\}$.
\end{lemma}

\begin{proof}
For any $z\neq0$, there exists a unique value of its canonical homogeneous norm: $r=\|z\|_\dn>0$. By the homogeneity of $\|\cdot\|_\dn$, one has
$\left\|e^{-G_\dn(k_{\mathfrak{d}}^{\varrho}a)}z\right\|_\dn=e^{-k_{\mathfrak{d}}^{\varrho}a}r$ for any $k_{\mathfrak{d}}^{\varrho}\in\mathbb{Z}$, then
\[
\varrho\le \left\|\mathfrak{d}\left(-k_{\mathfrak{d}}^{\varrho}a\right)x\right\|_\dn<\varrho e^{a} \quad\Leftrightarrow\quad 
\varrho\le e^{-k_{\mathfrak{d}}^{\varrho}a}r<\varrho e^{a}.
\]
Taking logarithms we conclude
\[
\ln\varrho\le -k_{\mathfrak{d}}^{\varrho}a+\ln r < a+\ln\varrho \quad \Leftrightarrow \quad \frac{\ln\tfrac{\varrho}{r}}{a}\le k_{\mathfrak{d}}^{\varrho}<1+\frac{\ln\tfrac{\varrho}{r}}{a}.
\]
The half-open interval on the right has length~$1$, so it contains exactly one integer; thus there is a unique $k_{\mathfrak{d}}^{\varrho}\in\mathbb{Z}$  belonging to this interval and 
given by
$k_{\mathfrak{d}}^{\varrho}=\lceil a^{-1}\ln \frac{\rho}{r}\rceil$.
\end{proof}

The above lemma establishes an analog of a homogeneous projection for the discrete dilation. 
Specifically, it shows that for any vector $x \in \mathbb{R}^n \setminus \{\mathbf{0}\}$, 
there exists a unique discrete scaling $\ds(-k_{\ds}^{\varrho}(x)a)$ that uniquely projects $x$
onto the compact set $\Omega_{\ds}(\varrho)$.
Since
\[ 
\Omega_{\ds}(\varrho) = \ds(\ln \varrho)\,\Omega_{\ds}(1)
\]
 we select $\varrho = 1$ in order to simplify  the subsequent analysis.
For shortness, we denote, $\Omega_{\ds}=\Omega_{\ds}(1)$ and $k_\ds=k_\ds^1$.

\begin{remark}
  We refer to the discrete scaling operator $\Pi_{\Omega_{\mathfrak{d}}}(x) = \mathfrak{d}(-k_{\mathfrak{d}}(x)a)$ as the projector  to the set $\Omega_{\mathfrak{d}}$. The vector $z = \mathfrak{d}(-k_{\mathfrak{d}}(x)a)x$ is the projection of $x$ onto the set $\Omega_{\mathfrak{d}}$. The integer $k_{\mathfrak{d}}(x)$ is called below by the projection index. 
\end{remark}

For positive definite $\ds$-homogeneous functions, the following lemma holds. 
\begin{lemma}\label{lem:comparsion_FF}
Let $\mathfrak{d}$ be a linear discrete dilation in $\R^n$.
Let $F_1, F_2:\R^n\mapsto \R$ be continuous real-valued functions on $\mathbb{R}^n$, $\mathfrak{d}$-homogeneous of degree $\nu_1>0$ and $\nu_2>0$ , respectively, and $F_1$ is positive definite. 
Then,  for every $x\in\mathbb{R}^n$, the following holds:
\[
    c_{\min} \cdot F_1^{\nu_2/\nu_1}(x) \le F_2(x) \le c_{\max}\cdot F_1^{\nu_2/\nu_1}(x),
    \]
    where the constants $c_{\min}\in \R$ and $c_{\max}\in \R$ are as follows
    \[
    c_{\min} \!:=\!\! \inf_{z \in \Omega_\mathfrak{d}} \left( \tfrac{F_2(z)}{F_1^{\nu_2/\nu_1}(z)} \right), \quad c_{\max} \!:=\!\! \sup_{z \in \Omega_{\mathfrak{d}}} \left( \tfrac{F_2(z)}{F_1^{\nu_2/\nu_1}(z)} \right).
    \]
\end{lemma}
\begin{proof}
For the case that $x=\boldsymbol{0}$, since $F_1$ and $F_2$ are homogeneous of positive degree, both functions have zero at the origin, the inequality then holds.

For the case that $x\neq \boldsymbol{0}$, by the definition of a fundamental domain, there exists a unique integer 
\(k_\mathfrak{d}(x) \in \mathbb{Z}\) such that the dilated vector 
\[
    y = \mathfrak{d}(-k_\mathfrak{d}(x) a)\, x
\]
lies within the compact set $\Omega_\ds$. Taking into account
\begin{equation}    \begin{aligned}
        &F_1(y) = F_1(\mathfrak{d}(-k_\mathfrak{d}(x) \cdot a)x) = e^{\nu_1 k_\mathfrak{d}(x) a} F_1(x),\\&F_2(y) = F_2(\mathfrak{d}(-k_\mathfrak{d}(x) \cdot a)x) = e^{\nu_2 k_\mathfrak{d}(x) a} F_2(x).
   \end{aligned}
\end{equation}
we derive
\[
F_2(x)=\left(\tfrac{F_2(y)}{F_1^{\nu_2 / \nu_1}(y)}\right) F_1^{\nu_2 / \nu_1}(x).
\]
Then, since $\Omega_\ds$ is a compact and the functions $F_1$ and $F_2$ are continuous then,  by the Extreme Value Theorem, minimum and maximum values $c_{\min}$ and $c_{\max}$ are finite. % The proof is completed.
\end{proof}
A $\mathfrak{d}$-homogeneity of mapping is inherited by its derivatives.
\begin{lemma}\label{lem:discrate_partial_d}
    Let $\mathfrak{d}$ be a  linear discrete dilation in $\R^n$.
Let  a function $h\in C^1(\mathbb{R}^n\setminus\{\boldsymbol{0}\}, \mathbb{R}$) and a vector field $g\in C^1(\mathbb{R}^n\setminus\{\boldsymbol{0}\}, \mathbb{R}^n)$ be $\mathfrak{d}$-homogeneous of degree $\mu\in \R$, then 
		\begin{equation}
		\left.\frac{\partial h(z)}{\partial z}\right|_{z=\mathfrak{d}(s)x}\mathfrak{d}(s)=	e^{\mu s} \frac{\partial h}{\partial x} ,
		\end{equation}
			for all $x\in\mathbb{R}^n\setminus \{\boldsymbol{0}\}$ and $s\in\mathcal{S}$.
\end{lemma}

The proof is presented in  Appendix \ref{subsec:proof_lemma}.

Inspired by \cite{zubov1958systems}, \cite{rosier1992SCL}, \cite{nakamura2002:SICE}, we can extend the Zubov--Rosier theorem to systems that are homogeneous with respect to a discrete dilation. 
This provides a systematic tool for establishing stability properties of homogeneous systems with respect to discrete dilations.
\begin{theorem}\label{thm:exist_LF_discrete}
	Let an upper semi-continuous vector field $\tilde{F}:\mathbb{R}^n \rightrightarrows \mathbb{R}^n$ be nonempty-valued, compact-valued, convex-valued and  $\mathfrak{d}$-homogeneous  of degree $\mu \in \mathbb{R}$ with respect to a linear discrete dilation $\mathfrak{d}$ in $\R^n$.  Let $m>0$ be an arbitrary positive number. The system
\begin{equation}\label{eq:f(x)_discrete} 
		\dot{x} \in   \tilde{F}(x),
	\end{equation}
    is globally uniformly asymptotically stable if and only if there exists a positive definite, $\mathfrak{d}$-homogeneous function 
	$V:\mathbb{R}^n \mapsto [0,+\infty)$ of degree $m>0$ such that $V \in C(\mathbb{R}^n) \cap C^1(\mathbb{R}^n \backslash \{\boldsymbol{0}\})$,
	\[
	\sup _{z \in \tilde{F}(x)} \frac{\partial V(x)}{\partial x} z \leq -\rho V^{1+\tfrac{\mu}{m}}(x), \quad \forall x \neq \boldsymbol{0},
	\]
    where $\rho>0$ is some number.
	Moreover, the system is
	\begin{itemize}
		\item  globally uniformly finite-time stable for $\mu < 0$;
		\item  globally uniformly exponentially stable for $\mu = 0$;
		\item globally uniformly nearly fixed-time stable for $\mu > 0$.
	\end{itemize}
\end{theorem}

The proof is presented in Appendix~\ref{appendix:1}.

The above theorem shows that an asymptotically stable $\ds$-homogeneous system  admits a $\ds$-homogeneous Lyapunov function and retains the degree-dependent convergence rates (finite-time, exponential, or nearly fixed-time stability) similarly to the Zubov-Rosier theorem (see Theorem \ref{thm:exist_LF}). 

As mentioned in the Introduction,  existing control design methods use  continuous dilation, while considering a well-designed homogeneous control with quantization, the system can be symmetric only with respect to discrete dilation.
However, the $\dn$-homogeneous control system with quantization \eqref{eq:q_sys} is $\ds$-homogeneous if the quantizer $\mathfrak{q}(x)$ is $\ds$-homogeneous:
\[
\mathfrak{q}(\mathfrak{d}(s)x) = \mathfrak{d}(s)\mathfrak{q}(x), \ s\in \mathcal{S}. 
\]
Therefore, the design of a nonlinear control system with state quantization consist in a design of a $\mathfrak{d}$-homogeneous quantizer for a well-tuned $\dn$-homogeneous control system without state quantization.
This approach preserves homogeneity of the system in the case of state quantization, but the challenge is to ensure that stability is preserved as well. To address this issue, we introduce the concept of homogeneous sector-boundedness.

%% file: sector_boundedness.tex
Sector-boundedness is a classical concept in control theory used to describe nonlinearities bounded between two linear gains. It facilitates stability analysis through Lyapunov and passivity methods \cite{khalil2002:book}, {and it is widely used in control with quantization (see, e.g., \cite{Fu_etal_2005_TAC}, \cite{liu_jiang_etas2012Aut})}. However, the conventional sector-boundedness conditions may not be adequate for homogeneous control system since their scaling symmetry may differ from scaling symmetry of a linear system in Euclidean space. To address this issue, we introduce a homogeneous sector-boundedness condition which is consistent with the underlying geometry. For this purpose we use  homogeneous norms induced by linear continuous dilations. The mentioned condition enables proper analysis of nonlinear behavior  such as finite-time and fixed-time stability.

First of all, let us recall that the homogeneous norm $\|\cdot\|_\dn$ is  a norm in a vector space homeomorphic to $(\mathbb{R}^n, \|\cdot\|)$. 

\begin{proposition}\!\cite[Lemma 7.3]{polyakov2020book}\label{prop:Phi}
      Let a linear continuous dilation $\mathbf{d}$ be a strictly monotone dilation with respect to a norm $\|\cdot\|$. Let $\|\cdot\|_{\dn}$ be the canonical homogeneous norm induced by $\|\cdot\|$. The mapping $\Phi_\dn: \mathbb{R}^n \mapsto \mathbb{R}^n$ given by
\begin{equation}
\Phi_\dn(x)=\|x\|_{\mathbf{d}} \mathbf{d}\left(-\ln \|x\|_{\mathbf{d}}\right) x, \quad x \in \mathbb{R}^n,\label{eq:Phi_d}
\end{equation}
is homeomorphism on $\mathbb{R}^n$, its inverse is given by
\begin{equation}
\Phi_\dn^{-1}(z)=\|z\|^{-1} \mathbf{d}(\ln \|z\|) z, \quad z \in \mathbb{R}^n,\label{eq:Phi_d_1}
\end{equation}
and $\Phi_\dn(\mathbf{0})=\Phi_\dn^{-1}(\mathbf{0})=\mathbf{0}$ by continuity.
\end{proposition}

The following theorem introduces the normed vector space $(\mathbb{R}^n_{\dn}, \|\cdot\|_\dn)$  over the field of reals with new rules for addition of vectors in $\R^n$ and for multiplication of a vector by a scalar.
\begin{theorem}\cite[Theorem 7.1]{polyakov2020book}
    Let a linear continuous dilation $\mathbf{d}$ in $\mathbb{R}^n$ be strictly monotone with respect to a norm $\|\cdot\|$. Let the canonical homogeneous norm $\|\cdot\|_{\mathbf{d}}$ be induced by $\|\cdot\|$. Let an addition of vectors
$\tilde{+}: \mathbb{R}^n \times \mathbb{R}^n \mapsto \mathbb{R}^n$ 
and a multiplication by a scalar
$
\tilde{\cdot}: \mathbb{R} \times \mathbb{R}^n \mapsto  \mathbb{R}^n
$ 
be defined as 
\begin{itemize}
    \item  $x \ \tilde{+}\  y:=\Phi_\dn^{-1}(\Phi_\dn(x)+\Phi_\dn(y))$, where $x, y \in \mathbb{R}^n$,
\item  $\lambda \ \tilde{\cdot} \ x:=\operatorname{sign}(\lambda) \mathbf{d}(\ln |\lambda|) x$, where $\lambda \in \mathbb{R}, x \in \mathbb{R}^n$,
\end{itemize}
where $\Phi_\dn$ is given by \eqref{eq:Phi_d}. Then the set $\mathbb{R}^n$ together with the operations $\tilde{+}$ and $\tilde{\cdot}$ is a vector space $\mathbb{R}_{\mathbf{d}}^n$ with the norm $\|\cdot\|_{\mathbf{d}}$.
\end{theorem}
We denote the subtraction operation in $\mathbb{R}_\dn$ by $\tilde{-}$, i.e., $x \ \tilde{-} \ y := x \ \tilde{+} \ (-y)$. The inner product in $\mathbb{R}^n_\dn$ can be defined as follows.

\begin{theorem}\cite[Theorem 5.4]{polyakov2025:book_vol_I}
Let an inner product in $\R^n$ be defined  as $\langle x, y \rangle=x^{\top} Py$ with $0\prec P=P^{\top}\in\R^{n\times n}, x,y\in \R^n$.
Let a linear dilation $\dn$  be strictly monotone with respect to the norm $\|x\|=\sqrt{x^{\top}Px}$. 
   The mapping $\langle \cdot, \cdot \rangle_{\mathbf{d}} : \mathbb{R}^n \times \mathbb{R}^n \mapsto \mathbb{R}$,  
\begin{equation}
\langle x, y \rangle_{\mathbf{d}} =
\left\langle \Phi_\dn(x),\; \Phi_\dn(y) \right\rangle
\end{equation}
is an inner product on $\R^n_{\dn}$.
\end{theorem}
The above theorem proves that $\R^n_{\dn}$ is a  Euclidean space too. To distinguish $\R^n$ and $\R^n_{\dn}$, below the space $\R^n_{\dn}$ is called \textit{homogeneous Euclidean space}\cite{polyakov2025:book_vol_I}. The homogeneous inner product and the canonical homogeneous norm are, obviously, linked as follows: $\|x\|_\dn=\sqrt{\langle x,x\rangle_\dn}$. 

Given matrix $H \in \mathbb{R}^{n \times n}$, we define its action ("multiplication") on a vector $x \in \mathbb{R}^n_\mathbf{d}$ as follows
\[
H \ \tilde{\cdot} \ x := \Phi_\mathbf{d}^{-1}(H \Phi_\mathbf{d}(x)).
\]
This definition ensures compatibility with the homogeneous addition and scalar multiplication. In particular, one can be shown that for $H=\alpha I_n, \alpha\in \R$ it holds $H \ \tilde \cdot \ x=\alpha \ \tilde \cdot \ x$ and
  $(H_1+H_2)\, \tilde \cdot \, (x_1\, \tilde + \, x_2)\!=\!H_1 \, \tilde \cdot \, x_1 \, \tilde + \, H_2 \, \tilde \cdot \, x_1\, \tilde + \, 
  H_1 \, \tilde \cdot \, x_2 \, \tilde + \, H_2 \, \tilde \cdot \, x_2$
  for all $H_1, H_2 \in \R^{n\times n}$ and all $x_1,x_2\in \R^n_{\dn}$.

The classical sector-boundedness condition characterizes nonlinear functions that lie within the sector defined by some linear operators $K_1$ and $K_2$.  A function $\phi$ is sector bounded by linear mappings $K_1$ and $K_2$ if the variation of $\phi$ is constrained by these linear mappings. This condition can be expressed either through an inner product inequality or, equivalently, in norm form, both capturing the geometric sector bounded by $K_1$ and $K_2$.
\begin{definition}\cite{khalil2002:book}
Let $\phi : \mathbb{R}^n \mapsto \mathbb{R}^n$ be a vector-valued function, and let $K_1, K_2 \in \mathbb{R}^{n \times n}$ be symmetric matrices. The function $\phi$ is said to be sector-bounded in the sector $[K_1, K_2]$ if $K := K_2 - K_1$ is symmetric positive definite and  \[\langle \phi(x) - K_1 x, \phi(x) - K_2 x \rangle \leq 0, \quad \forall x \in \mathbb{R}^n.\]
\end{definition}

Analogously, the homogeneous sector-boundedness condition generalizes the classical concept by replacing the Euclidean norm and standard linear operations with their homogeneous counterparts, defined via dilation operators and homogeneous norms. 
\begin{definition}[Homogeneous sector-boundedness]
Let $\mathbf{d}$ be a linear continuous dilation in $\R^n$. Let $\R^n_{\dn} $ be a homogeneous Euclidean space.
 Let $\phi : \mathbb{R}^n \mapsto \mathbb{R}^n$ be a vector-valued function, and let $K_1, K_2 \in \mathbb{R}^{n \times n}$ be symmetric matrices.  
The function $\phi$ is $\dn$-homogeneous sector-bounded in the sector $[K_1, K_2]$ if $K := K_2 - K_1$ is positive definite and 
$$
\left\langle \phi(x)  \tilde{-}  K_1  \tilde{\cdot}  x,\   \phi(x) \tilde{-} K_2\tilde{\cdot}x \right\rangle_\dn\leq 0 ,\quad \forall x \in \mathbb{R}^n.$$
\end{definition}

This definition of homogeneous sector-boundedness is  consistent with the underlying homogeneous vector space structure induced by dilation $\dn$. According to the above definition, the homogeneous sector boundedness is equivalent to
\begin{equation}
    \left\langle \Phi_\dn(\phi(x)) -K_1 \Phi_\dn(x),  \Phi_\dn(\phi(x)) - K_2 \Phi_\dn(x) \right\rangle\le 0.
\end{equation}
It states that the nonlinear function $\phi(x)$ lies within the sector defined by the linear operators $K_1$ and $K_2$ acting through homogeneous scalar multiplication and addition (i.e., $K_1 \tilde{\cdot} x$ and $K_2 \tilde{\cdot} x$), with the deviation bounded relative to the homogeneous structure.

To illustrate the connection with the classical scalar sector-boundedness, consider the special case where
$
K_1 = L - \kappa I$,  $K_2 = L + \kappa I$,
for some linear operator $L \in \mathbb{R}^{n \times n}$ and scalar $\kappa > 0$. The classical sector condition reduces to
$$
\|\phi(x) - L x\| \leq \kappa \|x\|,
$$
representing a symmetric sector of width $2\kappa$ around the nominal linear map $Lx$. Analogously, the homogeneous sector-boundedness condition in this  case becomes
\[
\|\phi(x) \tilde{-} L \tilde{\cdot} x \|_{\mathbf{d}} \leq \kappa \|x\|_{\mathbf{d}},
\]
or equivalently,
\[
\|\Phi_{\mathbf{d}}(\phi(x)) - L \Phi_{\mathbf{d}}(x)\|\leq \kappa\|\Phi_\dn(x)\|.\]
In particular, when $L = I$, the homogeneous sector-boundedness condition can be simplified to
$$
(1 - \kappa) \|x\|_{\mathbf{d}} \leq \|\phi(x)\|_{\mathbf{d}} \leq (1 + \kappa) \|x\|_{\mathbf{d}},
$$
indicating that $\phi(x)$ remains within a scaled sector relative to $x$ in the homogeneous norm.

The classical sector-boundedness is a special case of homogeneous sector-boundedness when the dilation is standard $\dn(s)=e^s I_n$. In the homogeneous setting, the same structural inequality is preserved under the transformation $\Phi_\dn(\cdot)$, which is encoded by the linear continuous dilation. Thus, homogeneous sector-boundedness can be viewed as a dilation-inspired generalization of the classical sector-boundedness. It is suitable for analyzes of  systems where homogeneity plays a central role.

In Section~\ref{sec:condition}, we discover that the key difficulty in the stability analysis of the quantized control system \eqref{eq:q_sys} is to characterize a relation between vector distances in the conventional Euclidean space $\mathbb{R}^n$ and those measured in the homogeneous Euclidean space $\mathbb{R}^n_\dn$. The following lemma derives the required relation.

\begin{lemma}\label{lem:distance}
    Let $\dn$ be a strictly monotone dialation.
    Then there exist  $\alpha_1, \alpha_2\in \mathcal{K} $ such that
    \[
    \alpha_1\left(\!\frac{\langle y\tilde{-}x, y\tilde{-}x\rangle_\dn}{\langle x, x\rangle_\dn}\!\right)
    \!\le\! \frac{\langle y\!-\!x, y\!-\!x\rangle_\dn}{\langle x, x\rangle_\dn} \!\le\! \alpha_2\left(\!\frac{\langle y\tilde{-}x, y\tilde{-}x\rangle_\dn}{\langle x, x\rangle_\dn}\!\right),
    \]
    for all 
     $x,y\in \mathbb{R}^n\setminus\{\boldsymbol{0}\}$.
\end{lemma}
The proof is presented in Appendix~\ref{appendix:2}.

This lemma plays a fundamental role in characterization of a ``consistency'' of distances  measured in the standard Euclidean space $\R^n$ and in the homogeneous Euclidean space $\R^n_{\dn}$. 
The lemma is particularly useful for quantizer design and analysis of quantization error directly in the homogeneous Euclidean space. 
In the following, we address the stability problem with quantized measurements by using homogeneous sector-boundedness.

%% file: main_results.tex
\section{Homogeneous stabilization with quantization}\label{sec:condition}

The scaling symmetry inherent to systems with a continuous dilation provides a powerful analytic tool: it allows stability analysis to be conducted on any compact subset of the state space that excludes the origin and then extended to the entire space via homogeneity. This scaling property on the state space greatly simplifies the stability proof. In this section, we establish the analogous results for systems defined by a discrete dilation.

Since any $\dn$-homogeneous system is  $\mathfrak{d}$-homogeneous too, we begin our analysis by assuming that the closed-loop system \eqref{eq:f} is $\mathfrak{d}$-homogeneous.

\begin{theorem}\label{thm:main}
Let $\dn$ be a linear continuous dilation in $\R^n$, and let $\mathfrak{d}:\mathcal{S}\mapsto \mathbb{R}^{n\times n}$ be the corresponding linear discrete dilation in $\R^n$.
Let the canonical homogeneous norm $\|x\|_{\dn}$ be induced by the weighted Euclidean norm $\|x\| = \sqrt{x^\top P x}$ with $P\succ 0$.
Let the system \eqref{eq:f} be $\mathfrak{d}$-homogeneous of degree $\mu$ and let Assumption~\ref{assmp:1} hold.  Let $\mathfrak{q}: \R^n \mapsto \mathcal{Q}$ be a quantizer.
If there exists a sufficiently small constant $\epsilon > 0$ such that
\begin{equation}\label{eq:condition}
    \|\mathfrak{q}(x) - x\|_\dn \le \epsilon \|x\|_\dn,\quad \forall x \in \mathbb{R}^n,
\end{equation}
then the quantized system \eqref{eq:q_sys} is globally asymptotically stable. Moreover, it is
\begin{itemize} 
    \item \textit{globally uniformly finite-time stable for $\mu<0$};
    \item \textit{globally uniformly exponentially stable for $\mu=0$}; 
    \item \textit{globally uniformly nearly fixed-time stable for $\mu>0$}.
\end{itemize}
\end{theorem}
The proof is presented in Appendix~\ref{appendix:3}.

It is not difficult to validate,  
that the sector-boundeness of the quantization error $\|\mathfrak{q}(x) - x\|\le\epsilon\|x\|$
 with a sufficiently small  $\epsilon>0$ preserves the  asymptotic stablility of linear control system (see, e.g., \cite{Fu_etal_2005_TAC}, \cite{bullo_Liberzon2006TAC},\cite{wang2021TAC}). 
For nonlinear homogeneous control systems, the same analysis can be based on a homogeneous sector-boundedness. 

{
The homogeneous sector-bounded condition provides a general stability criterion with a $\dn$-homogeneous system with state quantization. 
The quantized closed-loop system \ref{eq:q_sys}  does not need to be homogeneous. 
However, the discrete homogeneity of the system significantly simplifies the stability analysis.  

\begin{corollary}\label{coro:1}
Let $\ds: \mathcal{S} \mapsto \R^{n\times n}$ be a discrete linear dilation in $\R^n$ and
$
\Omega_\ds = \{z \in \mathbb{R}^n \mid 1 \le \|z\|_\dn <  e^a \}$ with $a>0$.
If the quantizer $\mathfrak{q}:\mathbb{R}^n \mapsto \mathcal{Q}$ satisfies
\[
\mathfrak{q}(\ds(s)x) = \ds(s)\mathfrak{q}(x), \quad \forall s\in \mathcal{S}, \quad \forall x\in \R^{n}, 
\]
then the system \eqref{eq:q_sys} is $\ds$-homogeneous. Moreover, if  there exists a sufficiently small $\epsilon > 0$ such that
\[
\|\mathfrak{q}(x)-x\|_\dn \le \epsilon, \quad \forall x \in \Omega_\ds,
\]
then the closed-loop system is globally asymptotically stable.
\end{corollary}
\begin{proof}
    Since for any $x\in \mathbb{R}^n\setminus\{\boldsymbol{0}\}$, there exists  $k_\ds:\mathbb{R}^n\mapsto \mathbb{Z}$ such that  $\ds(-k_\ds(x) a)x\in \Omega_{\ds}$, one has
    \[
    \|\mathfrak{q}(\ds(-k_\ds(x))x)-\ds(-k_\ds(x))x\|_\dn\le \epsilon , \ \forall x\in \mathbb{R}^n\setminus\{\boldsymbol{0}\}.
    \]
    Due to the homogeneity, the latter is equivalent to
    \begin{equation}\label{neq:coro_1}
        \|\mathfrak{q}(x)-x\|_\dn\le \epsilon e^{k_\ds(x)a}.
    \end{equation}
    On the other hand, from the definition of $\Omega_\ds$, we have 
    \begin{equation}\label{neq:coro:2}
        1\le\|\ds(-k_\ds(x)a)x\|_\dn< e^a \Rightarrow e^{k_\ds(x)a}\le \|x\|_\dn< e^{(k_\ds(x)+1)a}.
    \end{equation}
    From \eqref{neq:coro_1} and \eqref{neq:coro:2} we derive $\|\mathfrak{q}(x)-x\|_\dn \le \epsilon \|x\|_\dn$. We complete the proof applying Theorem \ref{thm:main}.
    \end{proof}

This result highlights a key benefit of designing a system that preserves discrete homogeneity: \textit{it allows a global stability analysis to be reduced to a local one}. 
%Due to the properties of discrete homogeneity, the stability condition only needs to be verified on a compact set $\Omega_\ds$. 
This repeats the conclusion obtained for continuous dilations \cite{bhat_etal_2005_MCSS}.
}
Therefore, the key challenge is to design a quantizer that fulfills the homogeneous sector-boundedness condition and,  simultaneously, preserves the discrete homogeneity of the closed-loop system.

\section{generalized homogeneous quantizer design}\label{sec:quantizer}

To preserve system discrete homogeneity, the quantizer have to incorporate the dilation  into its structure. This is a crucial step for applying homogeneity-based analysis tools to the closed-loop system.

To this end, we first introduce homogeneous coordinates based on Proposition \ref{prop:Phi}.
The vector $y=\Phi_{\dn}(x)$ defines the so-called \textit{homogeneous coordinates} of the vector $x$. Inspired by vector quantizer design using conventional polar-spherical coordinates 
 \cite{gu2014SCL}, \cite{wang2018Aut},  \cite{wang2021TAC}, we introduce the  homogeneous polar-spherical coordinates. The following definition is inspired by \cite{rosier1992SCL}, \cite{Parly1997:CDC},\cite{grune2000SIAM} and \cite[Chapter 10]{polyakov2025:book_vol_I}.
\begin{definition}[Homogeneous polar-spherical coordinates]
Let  $\dn$ be a linear continuous dilation in $\R^n$ and $\|x\|_{\dn}$ be the canonical homogeneous norm in induced by the weighted Euclidean norm $\|x\| = \sqrt{x^\top P x}$ with $P\succ 0$. Let $y = \Phi_\dn(x)$ denote $\dn$-homogeneous coordinates of a vector $x \in \mathbb{R}^n$. Let the vector $z \in \mathbb{R}^n$ be defined as follows
$$z = \Theta(y) :=\begin{bmatrix} \|y\| & \theta_1 & \theta_2 & \cdots & \theta_{n-1} \end{bmatrix}^\top,$$
where $0 \le \theta_i \le \pi$ for $i\in \overline{1,n-2}$, $0 \le \theta_{n-1} < 2\pi$ and  the angles $\theta_i$ are given by
\[\theta_{n-1} = \operatorname{atan2}(y_n, y_{n-1}), \   \theta_i = \operatorname{atan2}\left( \sqrt{\textstyle\sum_{j=i+1}^{n} y_j^2}, y_i \right),\]
with  $\operatorname{atan2}(\cdot, \cdot)$ being the two-argument arc-tangent function. The vector $z$ is referred to as the $\dn$-homogeneous polar-spherical coordinates of the vector $x$.
\end{definition}
The inverse transformation $\Theta^{-1}$ is defined as follows
 $$
 y_1 \!=\|y\|\!\cos\theta_1, y_n\!=\|y\|\!\!\prod_{j=1}^{n-1} \!\sin\theta_j, \ y_k \!=\|y\|\! \cos\theta_k \!\prod_{j=1}^{k-1} \!\sin\theta_j, 
$$
where $ 
 k\in \overline{2,n-1}
$.

Since $\|\Phi_\dn(x)\| = \|x\|_\dn\in \mathbb{R}$ and 
$\dn(-\ln\|x\|_\dn)x \in [0,\pi]^{n-2} \times [0,2\pi)$, 
the homogeneous polar–spherical coordinates provide a decomposition of any vector 
$x \in \mathbb{R}^n \setminus \{\boldsymbol{0}\}$ into its homogeneous norm $\|x\|_\dn$ and 
its projection onto the unit sphere, represented in spherical coordinates. 
This decomposition effectively maps the unbounded $n$-dimensional real space in the original coordinates 
to a product of a positive real line (for the norm) and bounded intervals (for the angular coordinates), 
which simplifies the design of the quantizer.

In the following sections, we will provide detailed insights into quantizer design and the estimation of quantization errors.

\subsection{Homogeneous vector quantizer} 
    Leveraging the homogeneous  polar-spherical coordinates, a homogeneous vector quantizer is defined as follows:
\begin{definition}
	Let a linear continuous dilation $\dn:\mathbb{R}\mapsto \mathbb{R}^{n\times n}$ be strictly monotone with respect to the weighted Euclidean norm $\|x\| = \sqrt{x^\top P x}$. Let $\mathfrak{q}_s : S^{n-1}(1) \mapsto \mathcal{Q}_s \subset S^{n-1}(1)$ be a spherical quantizer of the unit sphere and 
    $\mathfrak{q}_r : \mathbb{R}_+ \mapsto \mathcal{Q}_r \subset \mathbb{R}_+$ be a scalar quantizer. The quantizer $\mathfrak{q}_h : \mathbb{R}^n \mapsto \mathcal{Q} \subset \mathbb{R}^n$ given by $\mathfrak{q}_h(\mathbf{0})=\mathbf{0}$ and 
		\begin{equation}\label{eq:hom_quantizer}
			\mathfrak{q}_h(x) = \Phi_\dn^{-1}\left(\mathfrak{q}_r(\|\Phi_\dn(x)\|)
            \,\mathfrak{q}_s\left(\tfrac{\Phi_\dn(x)}{\|\Phi_\dn(x)\|}\right)\right), \quad x\neq \mathbf{0}
		\end{equation}
    is said to be 
    \begin{itemize}
    \item a homogeneous polar-spherical quantizer;
    \item a homogeneous spherical quantizer if $\mathfrak{q}_r\equiv 1$.
\end{itemize}
\end{definition}

According to Proposition \ref{prop:Phi}, since
        \[
        \|\Phi_\dn(x)\|=\|x\|_\dn, \ \frac{\Phi_\dn(x)}{\|\Phi_\dn(x)\|} = \dn(-\ln\|x\|_\dn)x.
        \]
The homogeneous quantizer can be represented as 
\begin{equation}
    \mathfrak{q}_h(x) = \dn(\ln\mathfrak{q}_r(\|x\|_\dn))\mathfrak{q}_s(\pi_\dn(x)),
\end{equation}
	where $\pi_\dn(x) = \dn(-\ln\|x\|_\dn)x$ is the homogeneous projector on the unit sphere.
	This implies that any value $\mathfrak{q}_h(x)$ can be obtained by a discrete scaling the quantization seeds on the unit sphere. 
   In addition, the following holds:
    \begin{equation}
\|\mathfrak{q}_h(x)\|_\dn = \mathfrak{q}_r(\|x\|_\dn),\
\pi_\dn(\mathfrak{q}_h(x))= \mathfrak{q}_s(\pi_\dn(x)).
    \end{equation}

The spherical quantizer is dilation-invariant, i.e.,
\[\mathfrak{q}_s(\pi_\dn(\dn(s)x)) = \mathfrak{q}_s(\pi_\dn(x)), \ \forall s\in \mathbb{R},  \ \forall x\in \mathbb{R}^n\setminus \{\boldsymbol{0}\}.\]
Since $\mathcal{S}\subset \mathbb{R}$, the above also holds for the discrete dilation. As the spherical quantizer is dilation-invariant, the dilation properties of $\mathfrak{q}_h$ are determined by the radial quantizer $\mathfrak{q}_r$. 

In the following, the radial quantizer $\mathfrak{q}_r$ is given by the well-known logarithmic quantizer, which is typically represented in the following standard form \cite{Fu_etal_2005_TAC}:\begin{equation}\label{eq:log_quantizer}
	\mathfrak{q}_{\log}(z)\!=\!
	\nu^i \xi_0, \ \text{if} \ z\in \mathcal{I}_i,\  i\in \mathbb{Z}, \ \mathfrak{q}_r(0)=0, 
\end{equation}
where $\mathcal{I}_i := \left[ \tfrac{\nu^i  }{1+\delta}\xi_0, \tfrac{\nu^i }{1-\delta}\xi_0 \right)$, $\nu \in(0,1)$ represents the quantization density and $\delta=(1-\nu) /(1+\nu)$ characterizes the quantization error.
A small $\nu$ (resp., a large $\delta$) implies a coarse quantization, but a large $\nu$ (resp., a small $\delta$) means a dense quantization.

\begin{lemma}
Let $\mathfrak{q}_h$ be the homogeneous quantizer defined in \eqref{eq:q_h_homogeneity}, 
and let its radial quantizer $\mathfrak{q}_r$ be the logarithmic quantizer \eqref{eq:log_quantizer} 
with parameter $\nu \in (0,1)$.  
Then, the closed-loop system \eqref{eq:q_sys} with $\mathfrak{q}=\mathfrak{q}_h$  
is $\ds$-homogeneous of degree $\mu$, where 
$\ds:\mathcal{S}_\nu \mapsto \mathbb{R}^n$ is the discrete linear dilation generated by an anti-Hurwitz matrix $G_\dn\in \mathbb{R}^{n\times n}$ and 
\[
\mathcal{S}_\nu = \{ e^{-k \ln \nu} \mid k \in \mathbb{Z} \}.
\]
\end{lemma}

\begin{proof}
For any $x \in \mathbb{R}^n$, there exists an $i \in \mathbb{Z}$ such that $\|x\|_\dn \in \mathcal{I}_i$.
Then, for any dilation index $s_j = j \ln \nu$, we have
\[
\|\mathfrak{d}(s_j)x\|_\dn = e^{-j\ln \nu}\|x\|_\dn \in \mathcal{I}_{i-j}.
\]
Since $q_s(\pi_\dn(x))$ is dilation-invariant, then, using \eqref{eq:q_log_1}, we get
\[
\mathfrak{q}_h(\ds(s_j)x) = \ds(i-j)\ds(\ln\nu)\mathfrak{q}_s(\pi_\dn(s)) = \ds(s_j)\mathfrak{q}_h(s).
\]
Therefore, the homogeneous quantizer $\mathfrak{q}_h$ satisfies
\begin{equation}\label{eq:q_h_homogeneity}
 	\mathfrak{q}_h(\ds(s_i)x) = \ds(s_i)\,\mathfrak{q}_h(x), \quad \forall s_i \in \mathcal{S},
\end{equation}
i.e., the homogeneity of \eqref{eq:q_sys} is preserved under the discrete dilation $\ds$. The proof is complete.
\end{proof}

The homogeneous polar-spherical quantizer $\mathfrak{q}_r$ with $\mathfrak{q}_r=\mathfrak{q}_{\log}$ can be represented as follows:
\begin{equation}\label{eq:q_log_1}
\mathfrak{q}_h(x) =
  \ds(s_i)\,\mathfrak{q}_s(\pi_\dn(x)), \  
  \text{if } \|x\|_\dn \in \mathcal{I}_i,
\end{equation}
where $s_i\in 
\mathcal{S}_\nu$.
The above equation clearly demonstrates that quantized values are obtained by a discrete scaling of quantization seeds from the unit sphere via discrete linear dilation.  The parameter $\nu \in (0,1)$ determines the density of quantization seeds along the radial direction, and simultaneously specifies a ``density'' of discrete homogeneity.  

As shown above, spherical quantization does not affect discrete homogeneity.
 For the design of the spherical quantizer $\mathfrak{q}_s: S^{n-1}(1) \mapsto \mathcal{Q}_s \subset S^{n-1}(1)$ as described in \cite{zhou2025:Aut}, we transform the unit vector $\frac{y}{\|y\|}$ into spherical coordinates $\Theta\left(\frac{y}{\|y\|}\right)$, and then divide each angle into equal intervals. The quantizer is designed as follows:
\begin{equation}\label{eq:q_s}
\mathfrak{q}_s\left(\tfrac{y}{\|y\|}\right)\! =\! P^{-\tfrac{1}{2}}\Theta^{-1}([1, q(\theta_{\pi_\dn,1}), \cdots, q(\theta_{\pi_\dn,n-1})]^\top), 
\end{equation}
where
$
[1, \theta_{\pi_\dn,1}, \theta_{\pi_\dn,2}, \cdots, \theta_{\pi_\dn,n-1}]^\top = \Theta\left(\tfrac{P^{\frac{1}{2}}y}{\|y\|}\right) 
$,
$
q(\theta_{\pi_\dn,k})\!=\!\lfloor\tfrac{\theta_{\pi_\dn,k}}{\Delta} \!+\! \tfrac{1}{2}\rfloor \Delta$, 
$k \!\in \overline{1,n-1}$ and the parameter  $0\le\Delta\le \pi$ defines a density of the spherical quantizer. 
\begin{remark}
For any polar-spherical quantizer, only the radial direction is divided into infinite countable set of disjoint intervals (e.g., using a logarithmic quantizer), while finite disjoint sets cover the sphere.
   Thus, the polar-spherical quantizer is a more efficient quantization method compared to any element-wise quantizer which introduces an infinite countable set of disjoint intervals along each axis.
\end{remark}

In this section, we construct a homogeneous logarithmic polar–spherical coordinate system using homogeneous coordinates, which ensures that the system \eqref{eq:q_sys} with $\mathfrak{q}(x)=\mathfrak{q}_h(x)$ is $\ds$-homogeneous. We then proceed to analyze the stability of the system.

\subsection{Quantization error estimate}
Firstly, quantization error analysis in generalized homogeneous polar spherical coordinates is straightforward. For the logarithmic quantizer, it has been shown in \cite{Fu_etal_2005_TAC} that the quantization error is sector bounded, that is \[|\mathfrak{q}_r(z) - z|\le \delta |z|, \ \forall z\in \mathbb{R}.\]
Besides, the relation between elements quantization and vector quantization errors 
has been proved in \cite{wang2021TAC}, \cite{zhou2025:Aut}.
\begin{proposition}\label{prop:Delta}
Let $\dn$ be a strictly monotone dilation, and let $\mathfrak{q}_s: S^{n-1}(1) \mapsto \mathcal{Q}_s \subset S^{n-1}(1)$ be the generalized homogeneous spherical quantizer defined as in \eqref{eq:q_s}. Then, Then, for any homogeneous coordinate $y =\Phi(x)\in \mathbb{R}^n \setminus \{\boldsymbol{0}\}$, the quantization error satisfies:
\begin{equation}
\left\|\mathfrak{q}_s\left(\tfrac{y}{\|y\|}\right) - \tfrac{y}{\|y\|}\right\| 
\le \sqrt{2 - 2\left(2\cos^{2(n-1)}\left(\tfrac{\Delta}{2}\right) - 1\right)}.
\end{equation}
\end{proposition}
For the logarithmic quantizer, the quantization intervals can be rewritten as
\[
\mathcal{I}_i = \left[\xi_0\frac{(1+\nu) \nu^i}{2}, \xi_0\frac{(1+\nu) \nu^i}{2} \cdot \frac{1}{\nu} \right), 
\]
then, for any fixed $i\in \mathbb{Z}$, with $\xi_0=\tfrac{2}{1+\nu}$, we have
\begin{equation}
    \begin{aligned}
        \mathcal{I}_0=\Omega_\ds(1)
        := \{\, z \!\in\! \mathbb{R}^n \!\setminus\! \{\boldsymbol{0}\} \mid 1 \le \|z\|_{\dn} <  e^{-\ln\nu} \,\}.
    \end{aligned}
\end{equation}
It follows from Corollary \ref{coro:1} that the quantization error in homogeneous coordinates need only be evaluated locally on $\mathcal{I}_0$. Hence we obtain the following error estimate. 

\begin{lemma}\label{lem:kappa}
Let a linear continuous dilation $\mathbf{d}$ be  strictly monotone dilation with respect to the weighted Euclidean norm $\|x\| = \sqrt{x^\top P x}$, $x\in \mathbb{R}^n$, $P\succ 0$.
Let $\mathfrak{q}_r:\mathbb{R}^n\mapsto \mathcal{Q}_r\subset \mathbb{R}$ be a logarithmic quantizer $\mathfrak{q}_{\rm log}$ with parameter $\nu$ and $\xi_0 = \tfrac{2}{1+\nu}$, and let
$\mathfrak{q}_s: S^{n-1}(1) \to \mathcal{Q}_s \subset S^{n-1}(1)$ be the  spherical quantizer given by \eqref{eq:q_s}.
The quantization error of the homogeneous quantizer \eqref{eq:hom_quantizer} admits the estimate
\[
    \|\mathfrak{q}_h(x) \tilde{-} x\|_\dn \le \tilde{\epsilon}\,\|x\|_\dn, \qquad
    \tilde{\epsilon} := (1+\delta)\,\beta(\Delta) + \delta,
\]
where $\delta=(1-\nu) /(1+\nu)$ and
\[
    \beta(\Delta) = 2\sqrt{1-\cos^{2(n-1)}\!\big(\tfrac{\Delta}{2}\big)}.
\]
\end{lemma}
\begin{proof}
Due to the discrete homogeneity of $\mathfrak{q}_h$, it is sufficient to show the error for any $\|\xi\|_\dn\in \mathcal{I}_0$. For any $\|\xi\|_\dn\in \mathcal{I}_0$, we have
\[ \mathfrak{q}_r(\|\xi\|_\dn)=\tfrac{2}{1+\nu}=1+\delta, \ 
|\mathfrak{q}_r(\|\xi\|_\dn)-\|\xi\|_\dn |\le \delta.
\]
On the other hand, the $\mathfrak{q}_h$ on $\xi\in \mathcal{I}_0$ has
\begin{equation}
    \begin{aligned}
    \|\mathfrak{q}_h(\xi)\tilde{-} \xi\|_\dn=\|\mathfrak{q}_r(\|\xi\|_\dn)\mathfrak{q}_s(\pi_\dn(\xi)) -  \|\xi\|_\dn\pi_\dn(\xi)\|.
    \end{aligned}
\end{equation}
This can be bounded as
\begin{equation}
    \begin{aligned}        &\|\mathfrak{q}_r(\|\xi\|_\dn)\mathfrak{q}_s(\pi_\dn(\xi)) -  \|\xi\|_\dn\pi_\dn(\xi)\|\\
    &\le\!\|(\mathfrak{q}_r(\|\xi\|_\dn)\!-\!\|\xi\|_\dn)\pi_\dn(x)\| \!+\! \mathfrak{q}_r(\|\xi\|_\dn) \|\mathfrak{q}_s(\pi_\dn(\xi))\!-\! \pi_\dn(\xi)\|\\
    &\le \delta + (1+\delta)\beta(\Delta).
    \end{aligned}
\end{equation}
Then, the proof is complete.
\end{proof}

Since $\tilde{\epsilon}$ tends to zero as $\delta$ and $\Delta$ tend to zero.
The following corollary is straightforward.
\begin{corollary}
Let  a homogeneous polar-spherical quantizer $\mathfrak{q}_h$ be defined by \eqref{eq:hom_quantizer} with  $\mathfrak{q}_r$ given by \eqref{eq:log_quantizer}, $\xi_0=\tfrac{2}{1+\nu}$ and $\mathfrak{q}_s$ given by \eqref{eq:q_s}.
Then, there exist sufficiently small $\delta=\frac{1-\nu}{1+\nu}$ and $\Delta$ such that the system \eqref{eq:q_sys} with $\mathfrak{q}=\mathfrak{q}_h$ is globally asymptotically stable. Moreover, it is
\begin{itemize} 
    \item \textit{globally uniformly finite-time stable for $\mu<0$};
    \item \textit{globally uniformly exponentially stable for $\mu=0$}; 
    \item \textit{globally uniformly nearly fixed-time stable for $\mu>0$}.
\end{itemize}
\end{corollary} 
\begin{proof}
By Lemma~\ref{lem:distance} and Lemma~\ref{lem:kappa}, we have
$
\|\mathfrak{q}_h(x) - x\|_\dn \le \alpha_2(\tilde{\epsilon})\|x\|_{\dn}$.
The result then follows from Corollary~\ref{coro:1}.
\end{proof}

To better illustrate the design of a homogeneous polar-spherical quantizer, an example of a two-dimensional quantizer design is presented below. 
\begin{example}
Let $G_{\dn_1} = \left[\begin{smallmatrix} 1.5 & 0.6 \\ 0 & 1 \end{smallmatrix}\right]$ be a 
a generator of linear continuous dilation $\dn_1$, which defines the polar-spherical quantizer $\mathfrak{q}_{h}$ for $P=I_2$. 
The classical polar-spherical quantizer studied in \cite{gu2014SCL}, \cite{wang2018Aut}, and \cite{wang2021TAC} corresponds the standard dilation $\dn_2(s)=e^sI_2$ with the generator $G_{\dn_2} = I_2$. Figure~\ref{fig:2} illustrates the quantization cells and their corresponding seeds for polar-spherical quantizers having different dilation generators.
The figure highlighting how the geometric structure of the generator shapes the resulting quantizer.
\end{example}

\begin{figure}[h!]
	\centering
\includegraphics[width=0.35\textwidth]{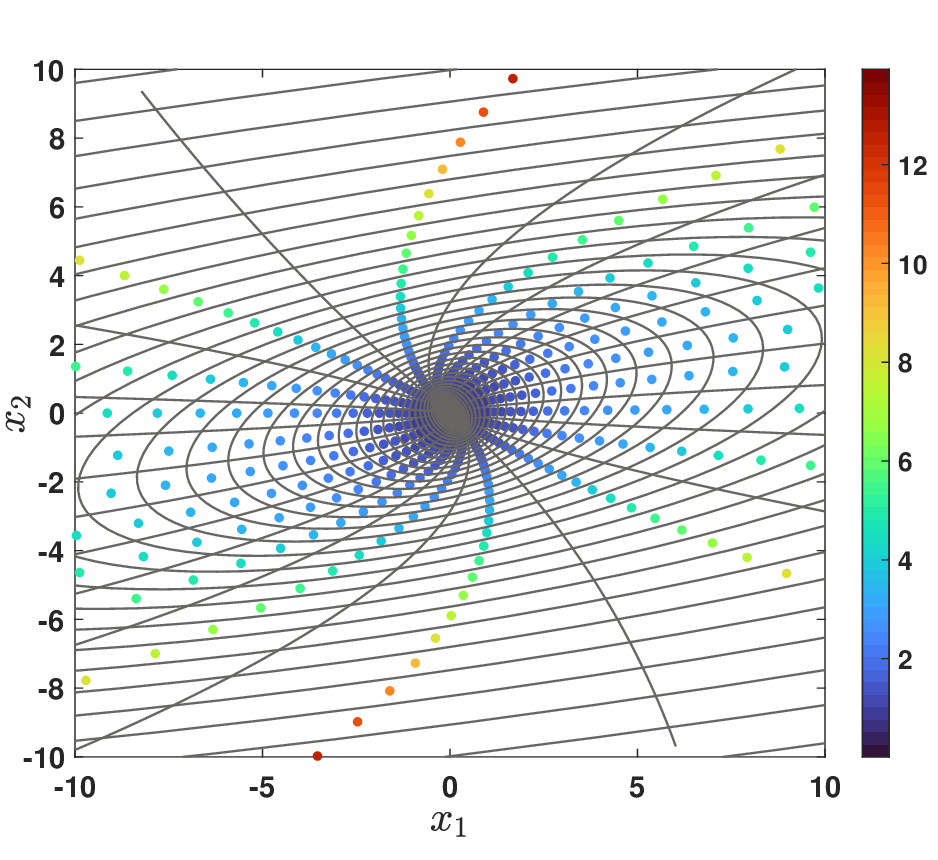}
\includegraphics[width=0.35\textwidth]{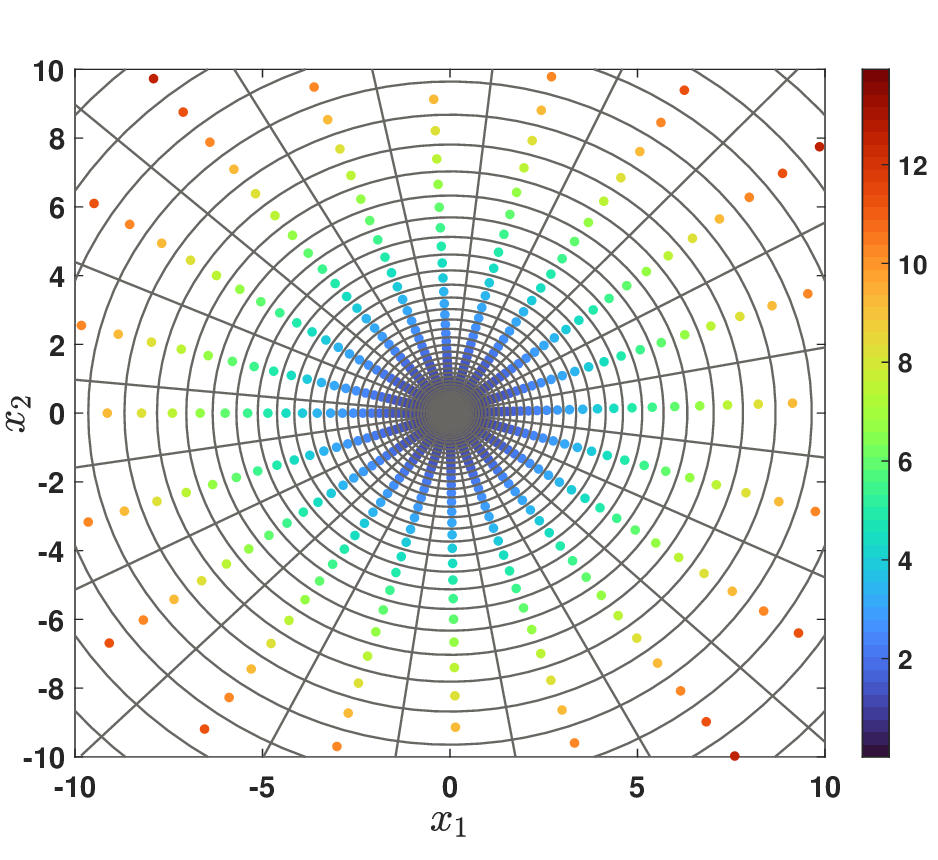}
	\caption{Quantization seeds under different generators: $G_{\dn_1}$ (top) and $G_{\dn_2}$ (bottom) with $P = I_2$. Colored points indicate quantization seeds, and lines outline the boundaries of quantization cells. The color bar represents the value of the homogeneous norm $\|\mathfrak{q}_h(x)\|_{\dn_i}$, $i=1,2$.}	\label{fig:2}
\end{figure}

%% file: simulation.tex
For numerical validation, we consider the following well-designed generalized homogeneous closed-loop system given in \cite{zimenko2023:Aut}:
\[
\dot{x} = \begin{bmatrix}
    x_2x_3^2 + x_2^2\\
    x_1\\x_2+x_3^2
\end{bmatrix}+Bu(x), \ u(x)= \|x\|_\dn^{4}K\dn(-\ln\|x\|_\dn)x.
\]
where 
$$B = \begin{bmatrix}
    1\\0\\0
\end{bmatrix},\quad  K = [-5.5055, -15.8387, -16.3807].$$
The state feedback closed-loop system is $\dn$-homogeneous with respect to the continuous linear dilation generated by $G_\dn = \left[\begin{smallmatrix}
    3 &0 &0\\
    0 & 2&0\\
    0 & 0 &1 
\end{smallmatrix}\right]$ of degree $\mu=1$, that is, the closed-loop system is nearly fixed-time stable.

We assume that the feedback control  has quantized state measurements $u(\mathfrak{q}_h(x))$, where the homogeneous quantizer has  parameters $\nu = 0.7$, $\Delta=\pi/20$. The corresponding simulation results are shown in Figure~\ref{fig:validation}. The proposed quantizer preserves the stability of the closed-loop system. However, due to the quantization error, the system exhibits a larger overshoot. 

As shown in Figure~\ref{fig:norm}, quantization leads to a slower convergence rate compared to the ideal (non-quantized) case. Moreover, in the right subfigure of Figure~\ref{fig:norm}, the logarithmic quantizer yields a uniform decay in the logarithmic scale, owing to the logarithmic quantization applied in the radial direction.
\begin{figure}[h!]
    \centering
    \includegraphics[width=0.95\linewidth]{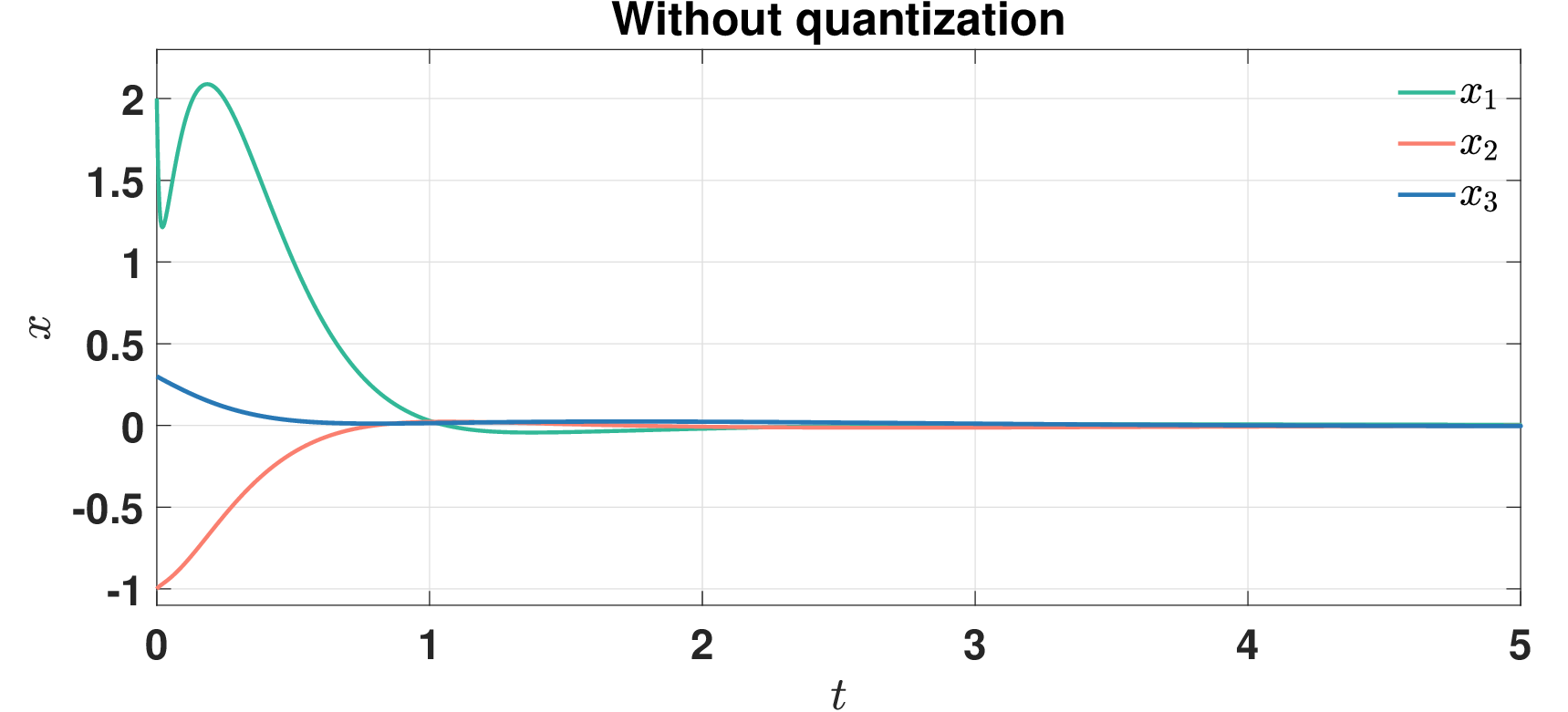}
    \includegraphics[width=0.95\linewidth]{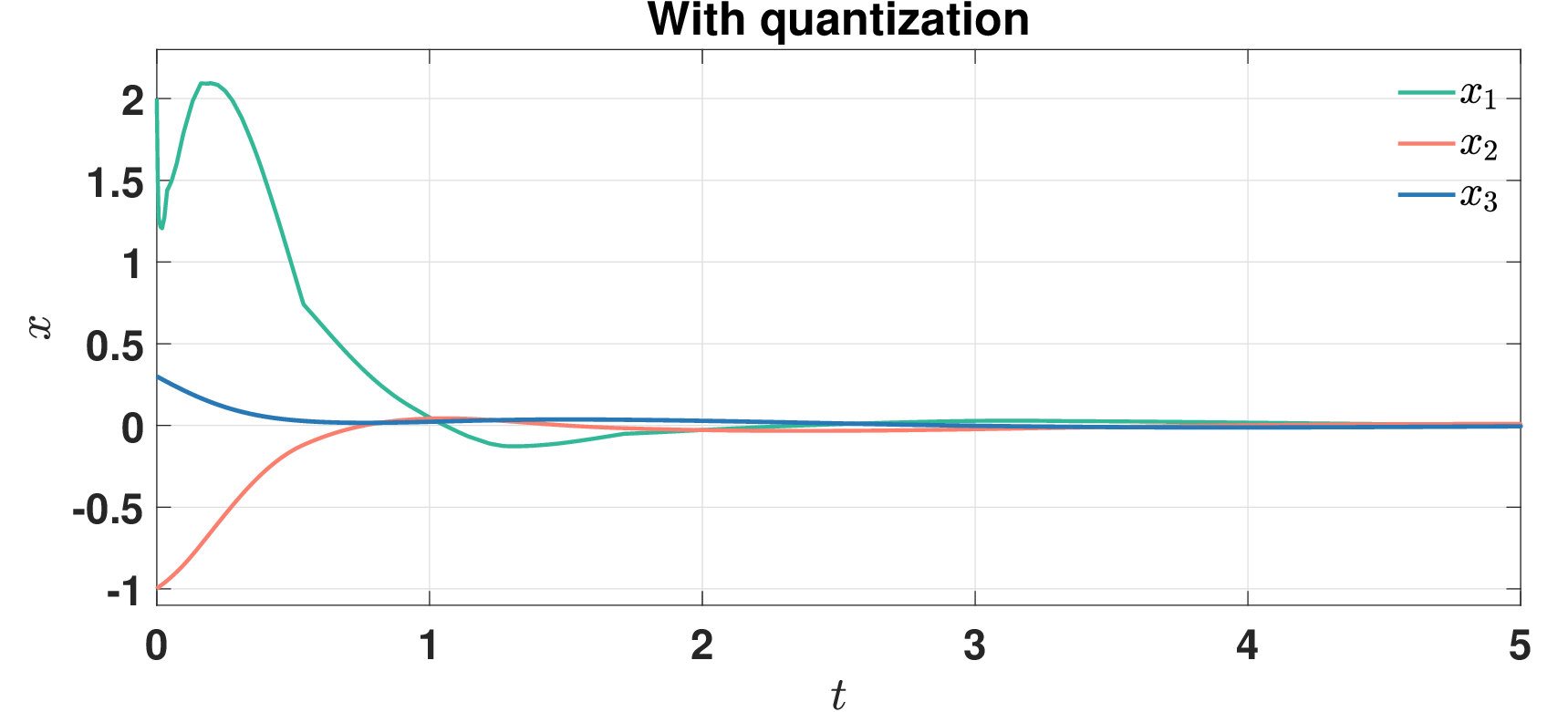}
    \caption{States  of the system under control with and without quantized data.}
    \label{fig:validation}
\end{figure}

\begin{figure}[h!]
    \centering
    \includegraphics[width=0.95\linewidth]{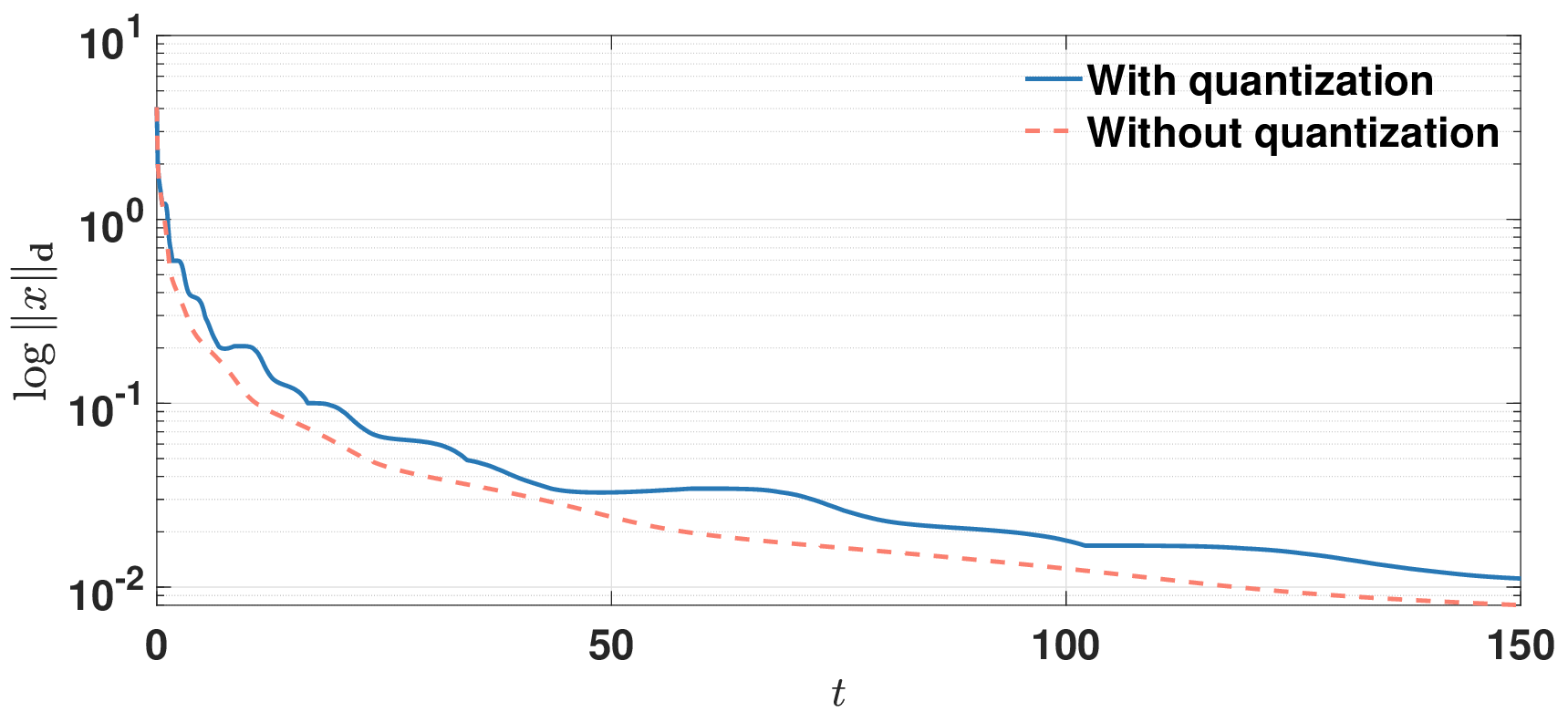}
    \includegraphics[width=0.95\linewidth]{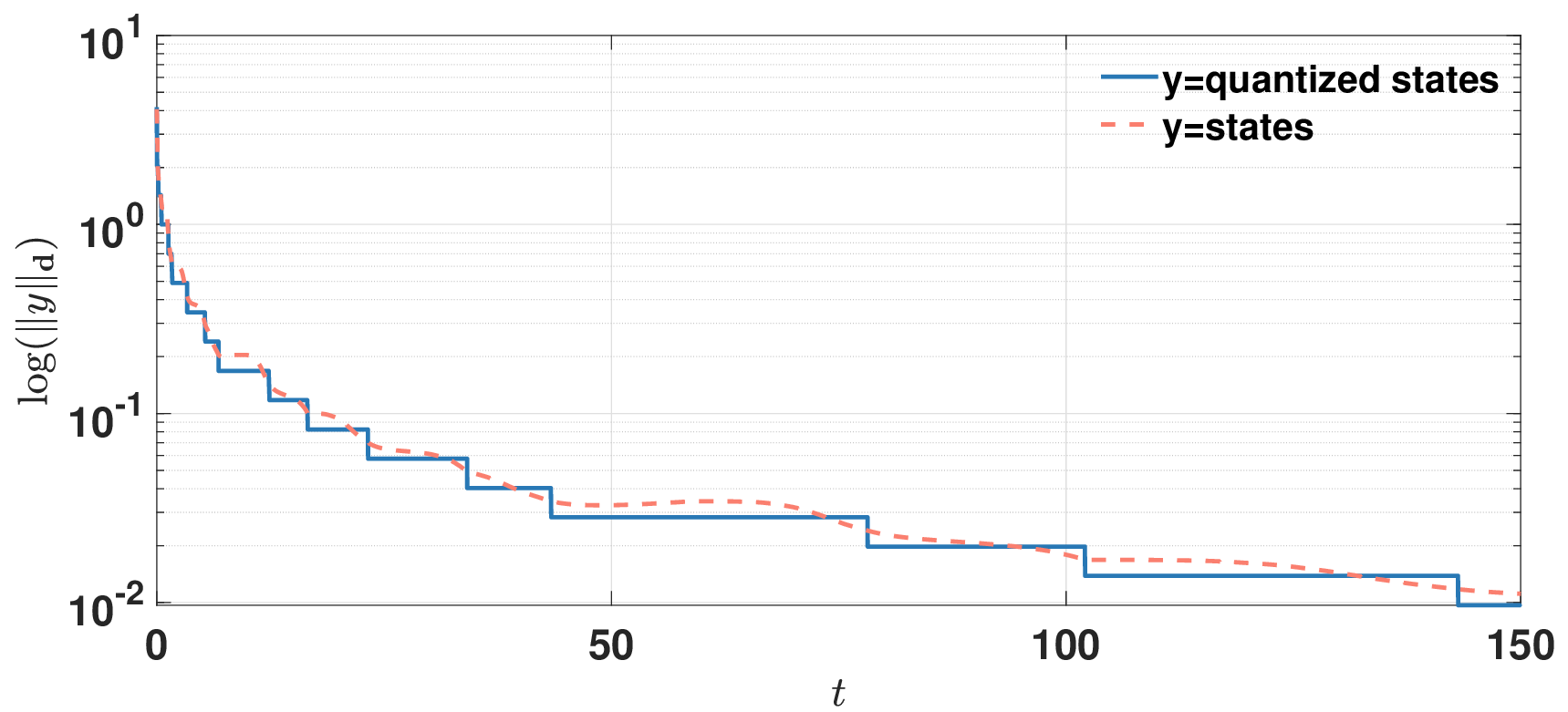}
    \caption{Comparison of state norms. The upper subfigure shows the norm of the system states under control with and without quantization. The lower subfigure illustrates the logarithmic norm of the states and the quantized states for the system with quantized state feedback.}
    \label{fig:norm}
\end{figure}

%% file: appendix.tex
\subsection{Proof of Lemma \ref{lem:discrate_partial_d}}\label{subsec:proof_lemma}
{According to the definition of derivative, we have
\begin{equation}
\lim_{\|y\|\to 0}
\frac{\bigl\lvert h(x+y)-h(x)-\tfrac{\partial h(x)}{\partial x}y \bigr\rvert}{\|y\|}
= 0,
\quad y \in \mathbb{R}^n,
\end{equation}
then, for $z = \mathfrak{d}(s)x$ and  $s\in \mathcal{S}$, using the $\mathfrak{d}$-homogeneity of $h$ we derive
\begin{equation}
   \begin{aligned}
        &\frac{
\bigl\lvert
h\!\left(\mathfrak{d}(s)x+y_s\right)
- h\!\left(\mathfrak{d}(s)x\right)
- \tfrac{\partial h(z)}{\partial z} y_s
\bigr\rvert
}{
\|y_s\|
}\\
&=
\frac{
e^{\mu s}
\bigl\lvert
h\!\left(x+\ds(-s)y_s\right)
- h(x)
- e^{-\mu s}\tfrac{\partial h(z)}{\partial z} y_s
\bigr\rvert
}{
\|y_s\|
} \\
&=
e^{\mu s}
\frac{
\bigl\lvert
h(x+y)
- h(x)
- e^{-\mu s}\tfrac{\partial h(z)}{\partial z} \mathfrak{d}(s)\mathfrak{d}(-s)y_s
\bigr\rvert
}{
\|\mathfrak{d}(s)\mathfrak{d}(-s)y_s
\|
}\\
&=
\frac{e^{\mu s}\|y\|}{\|\ds(s)y\|}
\frac{
\bigl\lvert
h(x+y)
- h(x)
- e^{-\mu s}\tfrac{\partial h(z)}{\partial z} \mathfrak{d}(s)y
\bigr\rvert
}{
\|y
\|
},
   \end{aligned}
\end{equation}
where $y = \mathfrak{d}(-s)y_s$.
For any fixed $s \in \mathcal{S}$, the term $\frac{e^{\mu s}\|y\|}{\|\dn(s)y\|}$
is uniformly bounded from above and from below, and  $\|y\|\to 0$ if and only if $\|y_s\|\to 0$.
Hence, using the uniqueness of the derivative of a differentiable function, we  complete the proof.
}

\subsection{Proof of Theorem \ref{thm:exist_LF_discrete}}
\label{appendix:1}

%\begin{theorem}\cite{clarke1998:JDE}
%\label{thm:clarke-converse}
%Consider the differential inclusion
%\[
%\dot x\in F(x),
%\]
%where $F:\R^n\rightrightarrows\R^n$ is upper semicontinuous, nonempty-valued, compact-valued, convex-valued, and $F(0)=\{0\}$. Then the origin is globally uniformly asymptotically stable if and only if there exist
%$
%V\in C^\infty(\mathbb{R}^n;\mathbb{R})$, 
%$W\in C^\infty\big(\mathbb{R}^n\setminus\{0\};(0,+\infty)\big)$,
%with $V$ proper and positive definite, such that
%\[
%\dot V(x) \le -W(x),\qquad \forall x\neq 0.
%\]
%\end{theorem}
%We now establish the Zubov-Rosier theorem for $\mathfrak{d}$-homogeneous systems.
\textbf{Sufficiency:} According to Lemma \ref{lem:comparsion_FF}, since $V$ is positive definition $\mathfrak{d}$ homogeneous of degree $m>0$, then, it is always well defined
\[
\underline{b} \|x\|_\dn^m\le V(x) \le \overline{b} \|x\|_\dn^m, 
\]
where $\underline{b}=\inf_{z \in \Omega_{\mathfrak{d}}} \left( \tfrac{V(z)}{\|z\|_\dn^{m}} \right)$, $\overline{b}=\sup_{z \in \Omega_{\mathfrak{d}}} \left( \tfrac{V(z)}{\|z\|_\dn^{m}} \right)$.

\textbf{Necessity:} According to \cite{clarke1998:JDE}, there exists a Lyapunov function that
$$
\dot{V}_0(x)=\sup _{z \in \tilde{F}(x)} \frac{\partial V_0(x)}{\partial x} z \leq-W_0(x), \quad \forall x \neq \mathbf{0} .
$$
Since $V_0$ is globally proper,  there exists $\underline{\rho}, \bar{\rho} \in \mathcal{K}_{\infty}$ such that
$$
\underline{\rho}\left(\|x\|_{\mathbf{d}}\right) \leq V_0(x) \leq \bar{\rho}\left(\|x\|_{\mathbf{d}}\right), \quad \forall x \in \mathbb{R}^n .
$$

Let a smooth scalar function $\omega\in \mathcal{C}^{\infty}$ be such that  
\begin{equation}
    \omega(\xi) = \left\{
\begin{aligned}
    0, \  & \xi \in [0,\underline{v}],\\
\omega^*, \  & \xi \in [\overline{v},\infty],
\end{aligned}    \right. \ 
\end{equation}
where $\omega^*>0$, $\omega'(\xi)>0$ for $\xi\in (\underline{v}, \overline{v})$ and 
\[
\omega^{(j)}(\underline{v}) = \omega^{(j)}(\overline{v}) = 0, \quad \forall j \geq 1.
\]
Let $r>0$ be an arbitrary number. Let $\bar{v}=\overline{\rho}(r)$ and let $r_0<r$ be such that $\underline{v}= \underline{\rho}\left(r_0\right)<\bar{v}$. Then, one has
$$
\underline{v} \leq V_0(x) \leq \bar{v} \Rightarrow r_0 \leq\|x\|_{\mathbf{d}} \leq r.
$$
Using $\omega$, we construct a new Lyapunov candidate
\[
\tilde{V}(x) = \sum_{i=-\infty}^{+\infty} e^{-m s_i} (\omega\circ V_0)(\mathfrak{d}(s_i)x), \ s_i = i \cdot a.
\]
Let us prove some properties of $\tilde{V}$.

\textbf{Well-definedness of $\tilde{V}$:}

For any $x\in \mathbb{R}^n\setminus\{\boldsymbol{0}\}$, according to the limit properties of discrete homogeneity, we have
\[
i\rightarrow -\infty \Rightarrow \|\mathfrak{d}(s_i)x\| \rightarrow 0, \ i\rightarrow +\infty \Rightarrow \|\mathfrak{d}(s_i)x\| \rightarrow +\infty.
\]
Since $V_0$ is continuous and positive definite, we have
\[
i\rightarrow -\infty \Rightarrow \tilde{V}(\mathfrak{d}(s_i)x) \rightarrow 0, \ i\rightarrow +\infty \Rightarrow \tilde{V}(\mathfrak{d}(s_i)x) \rightarrow +\infty.
\]
Thus, for any $x\neq \boldsymbol{0}$, we have 
\begin{equation}
    \begin{aligned}
        \underline{v}\le V_0(\mathfrak{d}(ka)x)\le \overline{v} &\Rightarrow r_0 \le \|\mathfrak{d}(ka)x\|_\dn =e^{ka}\|x\|_\dn\le r\\ &\Rightarrow k\in \left[\underline{k}, \overline{k}\right],
    \end{aligned}
\end{equation}
where $\underline{k}:=\left\lfloor\tfrac{1}{a}\ln\tfrac{r_0}{\|x\|_\dn}\right\rfloor$, $\overline{k}:=\left\lceil\tfrac{1}{a}\ln\tfrac{r}{\|x\|_\dn}\right\rceil$.

Then, the composite function has
\begin{equation}\label{eq:tranc}
    (\omega\circ V_0)(\mathfrak{d}(s_i)x)= \left\{
    \begin{aligned}
        & 0, & i\le \underline{k},\\
        & \omega^*, & i\ge \overline{k}.
    \end{aligned}
    \right.
\end{equation}
Then, for any $x\neq \boldsymbol{0}$, the function $\tilde{V}$ has
\begin{equation}\label{eq:tranc_V}
    \tilde{V}(x) =  \sum_{i=\underline{k}}^{\overline{k}-1} e^{-m s_i} (\omega\circ V_0)(\mathfrak{d}(s_i)x) + \sum_{i=\overline{k}}^{+\infty} \omega^*\cdot (e^{-ma})^i.
\end{equation}
For any $m>0$, the sum $\sum_{i=\overline{k}}^{\infty} e^{-m s_i}$ is a convergent geometric series.
Then, for any $x\in \mathbb{R}^n\setminus\{\boldsymbol{0}\}$, the function $\tilde{V}$ is always finite.

\textbf{Positive definiteness of $\tilde{V}$:}

Since $x=\boldsymbol{0}\Rightarrow \mathfrak{d}(s_i)x=\boldsymbol{0}$, $\forall s_i\in \mathcal{S}$, one has
$$V_0(\ds(s_i)x)=0, \ \forall s\in \mathcal{S}\Rightarrow \tilde{V}(\boldsymbol{0})=0.
$$
Besides,  $V_0(x) > 0$ for $x \neq \mathbf{0}$, we have \eqref{eq:tranc_V} holds, and
\[
e^{-m s_i} (\omega\circ V_0)(\mathfrak{d}(s_i)x) \geq 0, \ \forall i\in [\underline{k}, \overline{k}-1].
\]
{
On the other hand,
the   series
$\sum_{i=\overline{k}}^{+\infty} \omega^*\cdot (e^{-ma})^i>0$ is always positive, then $\tilde{V}(x)$ is also always positive away from the origin.

Thus, $\tilde{V}(x)$ is positive definite for any $x\in \mathbb{R}^n\setminus\{\boldsymbol{0}\}$ and $\tilde{V}(\boldsymbol{0})=0$.

\textbf{Discrete homogeneity of $\tilde{V}$:}

Let us show the $\ds$-homogeneity of $\tilde{V}$. For any $s_j\in \mathcal{S}$, one has
\begin{equation}\label{eq:discrete_tilde_V_1}
    \begin{aligned}
        \tilde{V}(\mathfrak{d}(s_j)x) &= \sum_{i=-\infty}^{+\infty} e^{-m s_i} (\omega\circ V_0)(\mathfrak{d}(s_i)\mathfrak{d}(s_j)x)\\
&=e^{m s_j}\sum_{i=-\infty}^{+\infty} e^{-m (s_i+s_j)} (\omega\circ V_0)(\mathfrak{d}(s_i+s_j)x).
    \end{aligned}
\end{equation}
Now, change the index $k=i+j$, the function $\tilde{V}$ is $\mathfrak{d}$-homogeneous of degree $m$,
\begin{equation}\label{eq:discrete_tilde_V_2}
    \begin{aligned}
        \tilde{V}(\mathfrak{d}(s_j)x) &=  e^{m s_j}\sum_{k=-\infty}^{+\infty} e^{-m (s_k)} (\omega\circ V_0)(\mathfrak{d}(s_k)x)\\
        &= e^{m s_j} \tilde{V}(x).
    \end{aligned}
\end{equation}
The constructed function $\tilde{V}$ is $\mathfrak{d}$-homogeneous of degree $m$.

\textbf{Radial unboundedness of $\tilde{V}$:}

For a fixed $x_0$ such that $\tilde{V}(x_0) = \tilde{C} > 0$, due to discrete homogeneity of $\tilde{V}$, one has
\[
\tilde{V}(\mathfrak{d}(s_j)x_0) = e^{m s_j} \tilde{V}(x_0) = \tilde{C} e^{m s_j}.
\] 
As the index $j \to \infty$, $s_j \to \infty$ (since $\mathcal{S}$ is an unbounded set of exponents), and thus:
\[
\lim_{s_j \to \infty} \tilde{V}(\mathfrak{d}(s_j)x_0) = \lim_{s_j \to \infty} \tilde{C} e^{m s_j} = \infty \quad (\text{since } m>0).
\]
According to the limitaion properties of discrete dilation,  the  $\mathfrak{d}(s_j)x_0$ tends to infinity in norm ($\|\mathfrak{d}(s_j)x_0\| \to \infty$) and the function value $\tilde{V}(\mathfrak{d}(s_j)x_0)$ also tends to infinity, the function $\tilde{V}(x)$ is indeed radially unbounded.

\textbf{Smoothness of $\tilde{V}$:}

For \eqref{eq:tranc_V}, let change the index $j=i-\underline{k}$,
\begin{equation}
    \begin{aligned}
        \tilde{V}(x) &= \sum_{j=0}^{\infty} e^{-m (j+\underline{k})a} (\omega\circ V_0)(\mathfrak{d}(ja +\underline{k}a )x)\\
    &=e^{-m\underline{k}a}\sum_{j=0}^{\infty} e^{-m (ja)} (\omega\circ V_0)(\mathfrak{d}(ja +\underline{k}a )x).
    \end{aligned}
\end{equation}
Since $ma>0$ and $e^{-ma}<1$, the above series converges. Taking $\underline{k}\le \tfrac{1}{a}\ln\tfrac{r_0}{\|x\|_\dn}$, we have
\[
\tilde{V}\le e^{-m\underline{k}a} \frac{\omega^*}{1-e^{-ma}}\le \frac{\|x\|_\dn^m}{r_0^m}\frac{\omega^*}{1-e^{-ma}}.
\]
Since $m>0$, then $\|x\|_\dn\rightarrow 0$ as $x\rightarrow \boldsymbol{0}$ and $\tilde{V}\ge 0$,  
The function $\tilde{V}$ is continuous at $\boldsymbol{0}$. 

Consider each 
\[
\tilde{V}_i = e^{-m s_i} (\omega\circ V_0)(\mathfrak{d}(s_i)x).
\]
Since $V_0$, $\omega$, and the flow $\mathfrak{d}(s_i)x = e^{s_i G_\dn} x$ are $C^{\infty}$, each term $\tilde{V}_i(x)$ is individually $C^{\infty}$.

According to \eqref{eq:tranc}, for any $x\in \mathbb{R}^n\setminus\{\boldsymbol{0}\}$, the derivative has
\begin{equation}
    \frac{\partial \tilde{V}}{\partial x} = \sum_{i= \underline{k}}^{\overline{k} -1} \frac{\partial \tilde{V}_i}{\partial x}.
\end{equation}
Since the function $\tilde{V}(x)$ is locally represented by a finite sum of functions, its $p$-th order derivative $\frac{\partial^p \tilde{V}}{\partial x^p}$ for any order $p \geq 1$ is simply the finite sum of the $p$-th order derivatives of the active terms:
$$\frac{\partial^p \tilde{V}}{\partial x^p} = \sum_{i= \underline{k}}^{\overline{k} - 1} \frac{\partial^p \tilde{V}_i}{\partial x^p}.
$$

Since each $\tilde{V}_i$ is $C^{\infty}$,  $\frac{\partial^p \tilde{V}_i}{\partial x^p}$ exists and is continuous for all $i$, the finite sum $\frac{\partial^p \tilde{V}_i}{\partial x^p}$ also exist and be continuous for any order $p$. The boundary conditions $\omega^{(j)}(\underline{v})=\omega^{(j)}(\overline{v})=0$ ensure that the derivatives transition smoothly as $x$ moves, locally changing the indices $\overline{k}$ and $\underline{k}$. Therefore, since the $k$-th derivative exists and is continuous for all $k$, $\tilde{V}(x)$ is infinitely differentiable  on $\mathbb{R}^n \setminus \{\boldsymbol{0}\}$.

\textbf{Time derivative of $\tilde{V}$:}

Using the homogeneity of $\tilde{F}$ and Lemma~\ref{lem:discrate_partial_d} on the partial derivatives of $\tilde{V}$,  the derivative of $\tilde{V}$ along \eqref{eq:f(x)_discrete} has
\begin{equation}
    \begin{aligned}
        \dot{\tilde{V}}(x) &= \sum_{i=-\infty}^{+\infty} e^{-m s_i}  (\omega'\circ V_0)(y)\cdot \sup _{z \in \tilde{F}(x)} \frac{\partial V_0(y)}{\partial y}\mathfrak{d}(s_i) z\\
        &=\sum_{i=-\infty}^{+\infty} e^{-(m+\mu)s_i} (\omega'\circ V_0)(y)\cdot \sup _{z \in \tilde{F}(y)} \frac{\partial V_0(y)}{\partial y} z.
    \end{aligned}
\end{equation}
where $y=\mathfrak{d}(s_i)x$.
For any $x\neq\boldsymbol{0}$, we have
\begin{equation}
    \dot{\tilde{V}} \le-W(x):= \sum_{i=-\infty}^{+\infty} e^{-(m+\mu)s_i} (\omega'\circ V_0)(y)\cdot W_0(y)<0.
\end{equation}
Similarly to \eqref{eq:discrete_tilde_V_1} and \eqref{eq:discrete_tilde_V_2}, one can show that $W(x)$ is $\mathfrak{d}$-homogeneous of degree $m+\mu$. Then, by Lemma \ref{lem:comparsion_FF}, there exists $\rho:=\inf_{z\in \Omega_{\mathfrak{d}}} \left(\tfrac{W(z)}{\tilde{V}^{1+\tfrac{\mu}{m}}(z)}\right)$, such that
\[
\rho \tilde{V}^{1+\tfrac{\mu}{m}}(x)\le W(x).
\]
Thus, one has
\[
\sup _{z \in \tilde{F}(x)} \frac{\partial \tilde{V}(x)}{\partial x} z\le -\rho \tilde{V}^{1+\tfrac{\mu}{m}}(x).
\]
Then, $\tilde{V}$ is a Lyapunov function for system \eqref{eq:f(x)_discrete}, the proof is accomplished.
}

\subsection{Proof of Lemma \ref{lem:distance}}\label{appendix:2}

From the homogeneous vector space operations, we deduce 
\begin{equation}
 \frac{\left\langle y-x, y-x\right\rangle_\dn}{\langle x, x\rangle_\dn}  = \frac{\|y-x\|_\dn^2}{\|x\|_\dn^2} \!=\! \|\dn(-\ln\|x\|_\dn)(y-x)\|_\dn^2,
\end{equation}
and 
\begin{equation}
    \frac{\langle y\tilde{-}x, y\tilde{-}x\rangle_\dn}{\langle x, x\rangle_\dn}=\frac{\|\Phi_\dn(y)-\Phi_\dn(x)\|^2}{\|\Phi_\dn(x)\|^2}.
\end{equation}

According to the relation between $\|\cdot\|_\dn$ and the Euclidean norm $\|\cdot\|$ in Proposition \ref{prop:diation_inequality}, one has
\begin{equation}\label{eq:d_inq1}
    \underbrace{\min \{\|x\|^{1/\overline{\eta}}, \|x\|^{1/\underline{\eta}}\}}_{:=\underline{\alpha}_1(\|x\|)}\le \|x\|_\dn \le\underbrace{\max\{\|x\|^{1/\overline{\eta}},\|x\|^{1/\underline{\eta}}\}}_{:=\overline{\alpha}_1(\|x\|)} .
\end{equation}
\noindent\textbf{Upper bound.}  
Let $s_x = \ln\|x\|_\dn$, $s_y = \ln\|y\|_\dn$, and $\tilde{s} = s_y - s_x$. We expand:
\begin{equation}\label{eq:d_inq2}
    \begin{aligned}
        &\|\dn(-\ln\|x\|)(y-x)\|\\
        &=\|\dn(-s_x)y - \dn(-s_y)y + \dn(-s_y)y - \dn(-s_x)x\|\\
        & =  \|\dn(\tilde{s})\dn(-s_y)y - \dn(-s_y)y + \dn(-s_y)y - \dn(-s_x)x\\
        & \le \left\|\dn(\tilde{s})-I_n\right\|+ \left\|\dn(-s_y)y - \dn(-s_x)x\right\|.
    \end{aligned} 
\end{equation}
Since $\dn(s) = e^{G_d s}$ and $\tfrac{d e^{G_d s}}{sd} = G_\dn e^{G_d s}$, we have
	\begin{equation}
		\begin{aligned}
			\|\dn(\tilde{s}) - I_n\| = \left\|\int_{0}^{\tilde{s}} G_\dn\dn(\tau)d\tau\right\|
			\!\le\!
			\|G_\dn\|\!\int_{0}^{\tilde{s}} \!\left\|\dn(\tau)\right\|d\tau.
		\end{aligned} 
	\end{equation} 
    Using the estimate in Proposition \ref{prop:diation_inequality}, it yields that
    \begin{equation}
            \int_{0}^{\tilde{s}} \left\|\dn(\tau)\right\|d\tau \le \left\{
        \begin{aligned}
    &\tfrac{1}{\overline{\eta}}(e^{\overline{\eta}\tilde{s}}-1), & \tilde{s}\ge 0,\\
    & \tfrac{1}{\underline{\eta}}(1-e^{\underline{\eta}\tilde{s}}), & \tilde{s}\le 0.
        \end{aligned}\right.
    \end{equation}
Moreover, the reverse triangle inequality \[\left| \|\Phi_\dn(y)\| - \|\Phi_\dn(x)\| \right| \le \|\Phi_\dn(y) - \Phi_\dn(x)\|,\] yields that
\begin{equation*}
\left\{
    \begin{aligned}
    & \|\Phi_\dn(x)\| - \|\Phi_\dn(y) - \Phi_\dn(x)\| \le \|\Phi_\dn(y)\|,\\
    &\|\Phi_\dn(y)\|\le \|\Phi_\dn(x)\| + \|\Phi_\dn(y) - \Phi_\dn(x)\|.
    \end{aligned}\right.
\end{equation*}

For $\|x\|_\dn>0$, the following holds 
\begin{equation}\label{eq:tilde_s_inq}
    1-\frac{\|\Phi_\dn(y)-\Phi_\dn(x)\|}{\|\Phi_\dn(x)\|}\le e^{\tilde{s}} \!=\! \frac{\|y\|_\dn}{\|x\|_\dn} \! \le\! \frac{\|\Phi_\dn(y)-\Phi_\dn(x)\|}{\|\Phi_\dn(x)\|}+1.
\end{equation}
Let us denote $\vartheta := \tfrac{\|\Phi_\dn(y) - \Phi_\dn(x)\|}{\|\Phi_\dn(x)\|}$.
Then we have
\begin{equation*}
    \begin{aligned}
        \left\|I_n-\dn\left(\tilde{s}\right)\right\|\le \overline{\alpha}_2(\vartheta):= \|G_\dn\|\max\left\{\tfrac{(\vartheta+1)^{\overline{\eta}}-1}{\overline{\eta}},\tfrac{1-(1-\vartheta)^{\underline{\eta}}}{\underline{\eta}} \right\}.
    \end{aligned}
\end{equation*}

On the other hand, observe that
\begin{equation}\label{eq:d_inq3}
    \begin{aligned}
        &\left\|\dn(-s_y)y - \dn(-s_x)x\right\|=\left\|\frac{\Phi_\dn(y)}{\|\Phi_\dn(y)\|}-\frac{\Phi_\dn(x)}{\|\Phi_\dn(x)\|}\right\|\\ 
        &= \left\|\frac{\Phi_\dn(y)}{\|\Phi_\dn(y)\|}-\frac{\Phi_\dn(y)}{\|\Phi_\dn(x)\|}+\frac{\Phi_\dn(y)}{\|\Phi_\dn(x)\|}-\frac{\Phi_\dn(x)}{\|\Phi_\dn(x)\|}\right\|\\
        &\le \left\|1-\frac{\Phi_\dn(y)}{\|\Phi_\dn(x)\|}\right\|+\left\|\frac{\Phi_\dn(y)}{\|\Phi_\dn(x)\|}-1\right\|\\ 
        &\le 2\frac{\|\Phi_\dn(y)-\Phi_\dn(x)\|}{\|\Phi_\dn(x)\|}.
    \end{aligned}
\end{equation}
Combining \eqref{eq:d_inq1}, \eqref{eq:d_inq2}, and \eqref{eq:d_inq3}, we conclude that
\[
 \frac{\left\langle y-x, y-x\right\rangle_\dn}{\langle x, x\rangle_\dn}\le\alpha_1(\vartheta):= \overline{\alpha}_1^2(\overline{\alpha}_2(\vartheta)+2\vartheta).
\]
\vspace{0.5em}
\noindent\textbf{Lower bound.}  
By the definition of  canonical homogeneous norm, we have:
\begin{equation}
    \begin{aligned}
        1=\|\dn(-s_x)x\|=\|\dn(-s_x)(y-x) + \dn(\tilde{s})\dn(-s_y)y\|.
    \end{aligned}
\end{equation}
Then, using triangle inequality, we have
\begin{equation}
    \begin{aligned}
        1&=\|\dn(-s_x)(x-y) + \dn(s_y-s_x)\dn(-s_y)y\|\\
        &\ge \left| \|\dn(-s_x)(y-x)\| - \|\dn(s_y-s_x)\dn(-s_y)y\|\right|.
    \end{aligned}
\end{equation}
The above inequality yields that
\[
1-\|\dn(s_y-s_x)\|\le \|\dn(-s_x)(y-x)\|.
\]
From Proposition \ref{prop:diation_inequality}, the bound of dilation yields
\[
1-\max\{e^{\overline{\eta}\tilde{s}},e^{\underline{\eta}\tilde{s}}\}\le\|\dn(-s_x)(y-x)\|.
\]
and reuse \eqref{eq:tilde_s_inq},
we have 
\[
\underbrace{1-\max\{(1+\vartheta)^{\overline{\eta}},(1+\vartheta)^{\underline{\eta}}\}}_{\underline{\alpha}_2(\vartheta)}\le\|\dn(-s_x)(y-x)\|
\]
Taking the \eqref{eq:d_inq1} into account, the lower bound function has  $\alpha_2:=\underline{\alpha}_1^2(\underline{\alpha}_2(\vartheta))$.
Then, the proof is completed.

\subsection{Proof of Theorem \ref{thm:main}}\label{appendix:3}

According to the Theorem \ref{thm:exist_LF_discrete}, there exists a $\mathfrak{d}$-homogeneous Lyapunov function $V \in C^{1}(\R^n\backslash\{ \boldsymbol{0}\}) \cap C(\R^n)$ of degree $m$ for the state feedback closed-loop system \eqref{eq:f} such that 
\begin{equation}\label{eq:V}
	\dot{V}(x) \leq -\rho V^{1+\tfrac{\mu}{m}}(x), \ \forall x\neq \boldsymbol{0}.
\end{equation}

Next, calculating the derivative of $V$ along with the quantization feedback system \eqref{eq:q_sys}, since $\mathfrak{q}:\mathbb{R}^n\mapsto \mathcal{Q}\subset \mathbb{R}$ and the discrete set $\mathcal{Q}$ is countable, we have the following holds almost everywhere:
\begin{equation}
	\begin{aligned}
		\dot{V}(x)  &
		= \frac{\partial V}{\partial x}\left[
		f(x)+g(x)u(x) + g(x)\tilde{u}
		\right] \\
		&\le -\rho V^{1+\frac{\mu}{m}} + \frac{\partial V(x)}{\partial x}g(x)\tilde{u},
	\end{aligned}
\end{equation}
where $\tilde{u}:=u(\mathfrak{q}(x))-u(x)$.

Given that $V(x)$ is $\dn$-homogeneous of degree $m$, by Lemma  \ref{lem:discrate_partial_d}, due to the partial derivative of a homogeneous function, one has:
\begin{equation}
		\begin{aligned}
		    \dot{V}(x)  
		\le&  -\rho V^{1+\tfrac{\mu}{m}}\\
		&+ e^{ m k_{\mathfrak{d}}(x)a}\left.\frac{\partial V(\xi)}{\partial \xi}\right|_{\xi = \ds(-k_{\mathfrak{d}}(x)a)x}\!\ds(-k_{\mathfrak{d}}(x)a)g(x)\tilde{u},
		\end{aligned}
\end{equation}
where $k_{\mathfrak{d}}(x)$ is the projection index to set $\Omega_{\mathfrak{d}}$.
 Due to homogeneity, one has
\begin{equation*}
    \begin{aligned}
        &
        g(\ds(-s_\mathfrak{d})x)\left[u(\ds(-s_\mathfrak{d})x)-u(\ds(-s_\mathfrak{d})\mathfrak{q}_x)\right]\\
        &=e^{-\mu s_{ \mathfrak{d}}}\ds(-s_\mathfrak{d})g(x)\tilde{u}
        ,
    \end{aligned}
\end{equation*}
where $s_\mathfrak{d} = k_{\mathfrak{d}}(x) a$.
Using the homogeneity of the system, it yields that:
\begin{equation}
	\begin{aligned}
		\dot{V}(x)  
		\le& -\rho V^{1+\tfrac{\mu}{m}}\\
		+& e^{(m+\mu)s_{\mathfrak{d}}}\!\!\left.\left[\frac{\partial V(\xi)}{\partial \xi}g(\xi)(u(\xi) - u(\xi+\sigma))\right]\right|_{\xi = \ds(-s_{\mathfrak{d}})x},
	\end{aligned}
\end{equation}
where $\sigma = \ds(-s_{\mathfrak{d}})(\mathfrak{q}(x) -x)$.

Due to the fact that 
\begin{equation}\label{neq:k_d}
    1\le \|\mathfrak{d}(-s_{\mathfrak{d}})x\|_\dn\le e^a \Leftrightarrow e^{s_{\mathfrak{d}}}\le \|x\|_\dn\le e^a e^{s_{\mathfrak{d}}}.
\end{equation}
Since $\mathfrak{d}(-s_{\mathfrak{d}})x$ lies on the set $\Omega_{\mathfrak{d}}$, and $\tfrac{\partial V(x)}{\partial x}$ is continuous on $\mathbb{R}^n$, then according to Extreme Value Theorem, there exists a positive constant $\overline{c}>0$ such that:
\begin{equation}\label{eq:neq2}
\left\|\left.\frac{\partial V(\xi)}{\partial \xi}\right|_{\xi = \mathfrak{d}(-s_{\mathfrak{d}})x}\right\| \le \overline{c}. 
\end{equation}
On the one hand, according to Lemma  \ref{lem:comparsion_FF}, we have:
\begin{equation}\label{neq:V_d}
	\underline{\rho} V^{\tfrac{1}{m}}\le \|x\|_\dn \le \overline{\rho} V^{\tfrac{1}{m}},  
\end{equation}
where
\[\underline{\rho} = \inf\limits_{\zeta\in \Omega_{\mathfrak{d}}}\left(\tfrac{\|\zeta\|_\dn}{V^{1/m}(\zeta)}\right), \ 
\overline{\rho} = \sup\limits_{\zeta\in \Omega_{\mathfrak{d}}}\left(\tfrac{\|\zeta\|_\dn}{V^{1/m}(\zeta)}\right).
\]
Taking \ref{neq:k_d} and \ref{neq:V_d} yields that
\[
e^{-a}\underline{\rho} V^{\tfrac{1}{m}}\le e^{s_\mathfrak{d} } \le \overline{\rho} V^{\tfrac{1}{m}}.
\]
On the other hand, the condition \eqref{eq:condition} implies that:
\begin{equation}
    \begin{aligned}
        &\|\mathfrak{q}(x)-x\|_\dn\le \epsilon\|x\|_\dn
        \Leftrightarrow 
        \|\ds(-s_{\mathfrak{d}})(\mathfrak{q}(x) -x)\|_\dn\le \epsilon e^a\\
        &\Leftrightarrow 
        \frac{1}{\epsilon e^a}\|\mathfrak{d}(-s_{\mathfrak{d}})(\mathfrak{q}(x) -x)\|_\dn\le 1\\
        &\Leftrightarrow 
        \|\dn(-\ln\epsilon (e^a))\mathfrak{d}(-s_{\mathfrak{d}})(\mathfrak{q}(x) -x)\|_\dn\le 1
        .
    \end{aligned}
\end{equation}
Since $\|x\|_\dn \!\le\! 1\Leftrightarrow \|x\|\!\le\! 1$, the above inequality becomes
\[
\|\mathfrak{d}(-s_{\mathfrak{d}})(\mathfrak{q}(x) -x)\|\le \frac{1}{        |\lfloor\dn(-\ln(\epsilon e^a))\rfloor|}.
\]
Then, according to \eqref{eq:d_bound}, one has
\[
\|\sigma\| = \|\mathfrak{d}(-s_{\mathfrak{d}})(\mathfrak{q}(x) \!-\!x)\| \le \max\{(\epsilon e^a)^{\underline{\eta}},(\epsilon e^a)^{\overline{\eta}} \}.
\]  
For all $x\neq \boldsymbol{0}$, there exists a compact set (which does not contain the origin), such that $\sigma$ always belongs to this compact set. Since $g$ and $u$ are continuous on the compact set, then they are uniformly continuous on the compact set, i.e., there exists a class-$\mathcal{K}$ function $\gamma$, such that:
\begin{equation}\label{eq:class_K}
	\begin{aligned}
		&\|g(\mathfrak{d}(-s_{\mathfrak{d}})x)\left[u(\mathfrak{d}(-s_{\mathfrak{d}})x)-u(\mathfrak{d}(-s_{\mathfrak{d}})\mathfrak{q}(x))\right]\|\\
		&\le \gamma\left( \|\sigma\|\right) \le \gamma\left(\max\{(a\epsilon)^{\underline{\eta}},(a\epsilon)^{\overline{\eta}} \}\right):=\overline{\gamma}(\epsilon).
	\end{aligned}
\end{equation}

Taking inequalities \eqref{neq:V_d}, \eqref{eq:neq2}, and \eqref{eq:class_K}, along with the derivative of $V$, we have:
\begin{equation}\label{eq:d_v_free}
	\dot{V}(x)  
	\le -\left(\rho- \tilde{\rho} \overline{c}\overline{\gamma}({\epsilon})\right)V^{1+\mu},
\end{equation}
almost everywhere, where $\tilde{\rho} = \max\{(e^{-a}\underline{\rho})^{1+\mu},\overline{\rho}^{1+\mu}\}$. Then, for a sufficiently small $\epsilon\le \overline{\gamma}^{-1}\left(\tfrac{\rho}{\tilde{\rho} \overline{c}}\right)$, the system is globally asymptotically stable. 

Then, the proof is completed.

%% file: reference.bib
@Article{rosier1992SCL,
	author    = {Rosier, L.},
	journal   = {Systems \& Control Letters},
	title     = {Homogeneous Lyapunov function for homogeneous continuous vector field},
	year      = {1992},
	number    = {6},
	pages     = {467--473},
	volume    = {19},
	publisher = {Elsevier},
}

@Article{zimenko_etal_2020_TAC,
	author    = {Zimenko, K. and Polyakov, A. and Efimov, D. and Perruquetti, W.},
	journal   = {IEEE Transactions on Automatic Control},
	title     = {Robust feedback stabilization of linear MIMO systems using generalized homogenization},
	year      = {2020},
	number    = {12},
	pages     = {5429--5436},
	volume    = {65},
	publisher = {IEEE},
}

@Article{bhat_etal_2005_MCSS,
	author    = {Bhat, S. P and Bernstein, D. S},
	journal   = {Mathematics of Control, Signals and Systems},
	title     = {Geometric homogeneity with applications to finite-time stability},
	year      = {2005},
	number    = {2},
	pages     = {101--127},
	volume    = {17},
	publisher = {Springer},
}

@article{elia_etal_2001_TAC,
	title={Stabilization of linear systems with limited information},
	author={Elia, N. and Mitter, S. K},
	journal={IEEE transactions on Automatic Control},
	volume={46},
	number={9},
	pages={1384--1400},
	year={2001},
	publisher={IEEE}
}

@article{Fu_etal_2005_TAC,
	title={The sector bound approach to quantized feedback control},
	author={Fu, M. and Xie, L.},
	journal={IEEE Transactions on Automatic control},
	volume={50},
	number={11},
	pages={1698--1711},
	year={2005},
	publisher={IEEE}
}

@InCollection{kawski1991,
	author    = {Kawski, M.},
	booktitle = {Analysis of controlled dynamical systems},
	publisher = {Springer},
	title     = {Families of dilations and asymptotic stability},
	year      = {1991},
	pages     = {285--294},
}

@Book{polyakov2020book,
	author    = {Polyakov, A.},
	publisher = {Springer},
	title     = {Generalized homogeneity in systems and control},
	year      = {2020},
}

@Article{polyakov2019IJRNC,
	author    = {Polyakov, A.},
	journal   = {International Journal of Robust and Nonlinear Control},
	title     = {Sliding mode control design using canonical homogeneous norm},
	year      = {2019},
	number    = {3},
	pages     = {682--701},
	volume    = {29},
	publisher = {Wiley Online Library},
}

@Article{kawski1990_CTAT,
	author  = {Kawski, M.},
	journal = {Control Theory and advanced technology},
	title   = {Homogeneous stabilizing feedback laws},
	year    = {1990},
	number  = {4},
	pages   = {497--516},
	volume  = {6},
}

@Article{hong2001output,
	author    = {Hong, Y. and Huang, J. and Xu, Y.},
	journal   = {IEEE Transactions on Automatic Control},
	title     = {On an output feedback finite-time stabilization problem},
	year      = {2001},
	number    = {2},
	pages     = {305--309},
	volume    = {46},
	publisher = {IEEE},
}

@Article{nakamura_etal_2009_TAC,
	author    = {Nakamura, N. and Nakamura, H. and Yamashita, Y. and Nishitani, H.},
	journal   = {IEEE Transactions on Automatic Control},
	title     = {Homogeneous stabilization for input affine homogeneous systems},
	year      = {2009},
	number    = {9},
	pages     = {2271--2275},
	volume    = {54},
	publisher = {IEEE},
}

@Article{zubov1958systems,
	author    = {Zubov, V. I.},
	journal   = {Izvestiya Vysshikh Uchebnykh Zavedenii. Matematika},
	title     = {Systems of ordinary differential equations with generalized-homogeneous right-hand sides},
	year      = {1958},
	number    = {1},
	pages     = {80--88},
	publisher = {Kazan (Volga region) Federal University},
}

@Article{khomenyuk1961systems,
	author    = {Khomenyuk, VV},
	journal   = {Izvestiya Vysshikh Uchebnykh Zavedenii. Matematika},
	title     = {Systems of ordinary differential equations with generalized homogeneous right-hand sides},
	year      = {1961},
	number    = {3},
	pages     = {157--164},
	publisher = {Kazan (Volga region) Federal University},
}

@book{filippov2013differential,
	title={Differential equations with discontinuous righthand sides: control systems},
	author={Filippov, A. F.},
	volume={18},
	year={2013},
	publisher={Springer Science \& Business Media}
}

@article{kawski1995IFAC,
	title={Geometric homogeneity and stabilization},
	author={Kawski, M.},
	journal={IFAC Proceedings Volumes},
	volume={28},
	number={14},
	pages={147--152},
	year={1995},
	publisher={Elsevier}
}

@Article{bullo_Liberzon2006TAC,
	author    = {Bullo, F. and Liberzon, D.},
	journal   = {IEEE Transactions on Automatic Control},
	title     = {Quantized control via locational optimization},
	year      = {2006},
	number    = {1},
	pages     = {2--13},
	volume    = {51},
	publisher = {IEEE},
}

@Article{andrieu2008homogeneous,
	author    = {Andrieu, V. and Praly, L. and Astolfi, A.},
	journal   = {SIAM Journal on Control and Optimization},
	title     = {Homogeneous approximation, recursive observer design, and output feedback},
	year      = {2008},
	number    = {4},
	pages     = {1814--1850},
	volume    = {47},
	publisher = {SIAM},
}

@Article{grune2000SIAM,
	author    = {Gr{\"u}ne, L.},
	journal   = {SIAM Journal on Control and Optimization},
	title     = {Homogeneous state feedback stabilization of homogenous systems},
	year      = {2000},
	number    = {4},
	pages     = {1288--1308},
	volume    = {38},
	publisher = {SIAM},
}

@Article{wang2021TAC,
  author    = {Wang, J.},
  journal   = {IEEE Transactions on Automatic Control},
  title     = {Quantized feedback control based on spherical polar coordinate quantizer},
  year      = {2021},
  number    = {12},
  pages     = {6077--6084},
  volume    = {66},
  publisher = {IEEE},
}

@Article{wang2018Aut,
  author    = {Wang, J.},
  journal   = {Automatica},
  title     = {Asymptotic stabilization of continuous-time linear systems with quantized state feedback},
  year      = {2018},
  pages     = {83--90},
  volume    = {88},
  publisher = {Elsevier},
}

@article{zhou2024:Tac,
  title={Finite/fixed-time stabilization of linear systems with state quantization},
  author={Zhou, Y. and Polyakov, A. and Zheng, G.},
  journal={IEEE Transactions on Automatic Control},
  year={2024},
  publisher={IEEE}
}

@Article{gu2014SCL,
  author    = {Gu, G. and Qiu, L.},
  journal   = {Systems \& Control Letters},
  title     = {Networked control systems for multi-input plants based on polar logarithmic quantization},
  year      = {2014},
  pages     = {16--22},
  volume    = {69},
  publisher = {Elsevier},
}

@article{swaszek1983multidimensional,
  title={Multidimensional spherical coordinates quantization},
  author={Swaszek, P and Thomas, J},
  journal={IEEE Transactions on Information Theory},
  volume={29},
  number={4},
  pages={570--576},
  year={1983},
  publisher={IEEE}
}

@book{khalil2002:book,
  title={Nonlinear systems},
  author={Khalil, H. K and Grizzle, J. W},
  volume={3},
  year={2002},
  publisher={Prentice hall Upper Saddle River, NJ}
}

@Article{zimenko2023:Aut,
  author    = {Zimenko, K. and Polyakov, A. and Efimov, D.},
  journal   = {Automatica},
  title     = {Homogeneous systems stabilization based on convex embedding},
  year      = {2023},
  pages     = {111108},
  volume    = {154},
  publisher = {Elsevier},
}

@Article{liu_jiang_etas2012Aut,
  author    = {Liu, T. and Jiang, Z-P and Hill, D. J},
  journal   = {Automatica},
  title     = {A sector bound approach to feedback control of nonlinear systems with state quantization},
  year      = {2012},
  number    = {1},
  pages     = {145--152},
  volume    = {48},
  publisher = {Elsevier},
}

@Article{ceragioli2007:SCL,
  author    = {Ceragioli, F. and De Persis, C.},
  journal   = {Systems \& control letters},
  title     = {Discontinuous stabilization of nonlinear systems: Quantized and switching controls},
  year      = {2007},
  number    = {7-8},
  pages     = {461--473},
  volume    = {56},
  publisher = {Elsevier},
}

@article{bikas2020:TAC,
  title={Tracking performance guarantees in the presence of quantization for uncertain nonlinear systems},
  author={Bikas, L. N and Rovithakis, G. A},
  journal={IEEE Transactions on Automatic Control},
  volume={66},
  number={7},
  pages={3311--3316},
  year={2020},
  publisher={IEEE}
}

@article{wang2021:TAC,
  title={Adaptive backstepping control of uncertain nonlinear systems with input and state quantization},
  author={Wang, W. and Zhou, J. and Wen, C. and Long, J.},
  journal={IEEE Transactions on Automatic Control},
  volume={67},
  number={12},
  pages={6754--6761},
  year={2021},
  publisher={IEEE}
}

@book{polyakov2025:book_vol_I,
	author = {Polyakov, A.},
	edition = {2nd},
	publisher = {Springer},
	title = {Generalized Homogeneity in Systems and Control Volume I: Finite-dimensional systems},
	year = {2025}}

@book{polyakov2025:book_vol_II,
	author = {Polyakov, A.},
	edition = {2nd},
	publisher = {Springer},
	title = {Generalized Homogeneity in Systems and Control Volume II: Infinite-dimensional systems},
	year = {2025}}

@InProceedings{Parly1997:CDC,
  author    = {Praly, L.},
  booktitle = {Proceedings of the 36th IEEE Conference on Decision and Control},
  title     = {Generalized weighted homogeneity and state dependent time scale for linear controllable systems},
  year      = {1997},
  pages     = {4342-4347 vol.5},
  volume    = {5},
  doi       = {10.1109/CDC.1997.649536},
}

@article{zhou2025:Aut,
  title={Robust finite-time stabilization of linear systems with limited state quantization},
  author={Zhou, Y. and Polyakov, A. and Zheng, G.},
  journal={Automatica},
  volume={171},
  pages={111967},
  year={2025},
  publisher={Elsevier}
}

@Article{grune2023:ifac,
  author    = {Gr{\"u}ne, L. and Worthmann, K.},
  journal   = {IFAC-PapersOnLine},
  title     = {Homogeneity for control systems in discrete time},
  year      = {2023},
  number    = {1},
  pages     = {385--390},
  volume    = {56},
  publisher = {Elsevier},
}

@article{sanchez2020:Aut,
  title={Discrete-time homogeneity: Robustness and approximation},
  author={Sanchez, T. and Efimov, D. and Polyakov, A.},
  journal={Automatica},
  volume={122},
  pages={109275},
  year={2020},
  publisher={Elsevier}
}

@InProceedings{nakamura2002:SICE,
  author       = {Nakamura, H. and Yamashita, Y. and Nishitani, H.},
  booktitle    = {Proceedings of the 41st SICE Annual Conference. SICE 2002.},
  title        = {Smooth Lyapunov functions for homogeneous differential inclusions},
  year         = {2002},
  organization = {IEEE},
  pages        = {1974--1979},
  volume       = {3},
}

@InProceedings{granzotto2021:CDC,
  author       = {Granzotto, M. and Postoyan, R. and Bușoniu, L. and Ne{\v{s}}i{\'c}, D. and Daafouz, J.},
  booktitle    = {2021 60th IEEE Conference on Decision and Control (CDC)},
  title        = {Exploiting homogeneity for the optimal control of discrete-time systems: application to value iteration},
  year         = {2021},
  organization = {IEEE},
  pages        = {6006--6011},
}

@article{clarke1998:JDE,
  title={Asymptotic stability and smooth Lyapunov functions},
  author={Clarke, F. H and Ledyaev, Y. S and Stern, R. J},
  journal={Journal of differential Equations},
  volume={149},
  number={1},
  pages={69--114},
  year={1998},
  publisher={Elsevier}
}

@article{rudin1976:book,
  title={Principles of mathematical analysis},
  author={Rudin, W.},
  journal={3rd ed.},
  year={1976}
}

@article{bhat2000:SIAM,
  title={Finite-time stability of continuous autonomous systems},
  author={Bhat, S. P and Bernstein, D. S},
  journal={SIAM Journal on Control and optimization},
  volume={38},
  number={3},
  pages={751--766},
  year={2000},
  publisher={SIAM}
}

@article{polyakov2011:TAC,
  title={Nonlinear feedback design for fixed-time stabilization of linear control systems},
  author={Polyakov, A.},
  journal={IEEE transactions on Automatic Control},
  volume={57},
  number={8},
  pages={2106--2110},
  year={2011},
  publisher={IEEE}
}
